\documentclass[12pt]{article}
\usepackage{a4,amsmath,amssymb,amsthm,amscd,cite,graphicx}
\usepackage{verbatim,numprint,siunitx,mathrsfs,esint,xcolor}
\usepackage{enumitem}
\usepackage{rotating}
\usepackage{hyperref} \hypersetup{colorlinks, linkcolor=red }
\usepackage{tikz}
\usetikzlibrary{intersections,calc,matrix,arrows, decorations.markings,shapes}
\usepackage{multirow}
\newtheorem{theorem}{Theorem}[section]
\newtheorem{prop}[theorem]{Proposition}
\newtheorem{conj}[theorem]{Conjecture}

\def\Bbb{\mathbb} \def\BZ{\Bbb Z}  
  \def\BH{\mathbb{H}}

 \newcommand{\be}{\begin{equation}}
  \newcommand{\ee}{\end{equation}}

\newcommand{\Mtwo}{M_{12}\!:\!2}
\newcommand{\wchi}{{\widehat{\chi}}}
\newcommand{\tchi}{{\widetilde{\chi}}}
\newcommand{\wrho}{{\widehat{\rho}}}
\newcommand{\trho}{{\widetilde{\rho}}}
\newcommand{\wpsi}{{\widehat{\psi}}}

\catcode`@=11 \@addtoreset{equation}{section} \catcode`@=12

\begin{document}
\bibliographystyle{utphys}

\begin{titlepage}
\begin{flushright}
 Oct 2018 (v2.1)\\
\texttt{arXiv:1804.06677 [hep-th]}
\end{flushright}
\begin{center}
\textsf{\large Two moonshines for $L_2(11)$ but none for $M_{12}$}\\[12pt]
Suresh Govindarajan$^{\dagger}$ and Sutapa Samanta$^{*}$ \\
Department of Physics\\
Indian Institute of Technology Madras \\
Chennai 600036 INDIA\\
Email: $^\dagger$suresh@physics.iitm.ac.in, $^*$sutapa@physics.iitm.ac.in
\end{center}
\begin{abstract}
In this paper, we revisit an earlier conjecture by one of us that related conjugacy classes of $M_{12}$ to Jacobi forms of weight one and index zero. 
We construct Jacobi forms for all conjugacy classes of $M_{12}$ that are consistent with  constraints from group theory as well as modularity.
However, we obtain 1427 solutions that satisfy these constraints (to the order that we checked) and are unable to provide a unique Jacobi form. Nevertheless, as a consequence, we are able
to provide a  group theoretic proof of the evenness of the coefficients of all EOT Jacobi forms associated with conjugacy classes of $M_{12}:2 \subset M_{24}$. We show that there exists no solution where the Jacobi forms (for order 4/8 elements of $M_{12}$) transform with phases under the appropriate level. In the absence of a moonshine for $M_{12}$, we show that there exist moonshines for two distinct $L_2(11)$ sub-groups of the $M_{12}$. We construct Siegel modular forms for all $L_2(11)$ conjugacy classes and show that each of them arises as the denominator formula for a distinct Borcherds-Kac-Moody Lie superalgebra. 
\end{abstract}
\end{titlepage}

\section{Introduction}

Following the discovery of monstrous moonshine, came a moonshine for the largest sporadic Mathieu group, $M_{24}$. This related $\rho$, a conjugacy class of $M_{24}$, to a multiplicative eta product that we denote by $\eta_\rho$ via the map\cite{Dummit:1985,Mason:1985}:
\begin{equation}
\rho=1^{a_1}2^{a_2}\cdots N^{a_N} \longrightarrow \eta_\rho(\tau):=\prod_{m=1}^N \eta(m\tau)^{a_m}\quad. 
\end{equation}
The same multiplicative eta products appeared as the  generating function of $\tfrac12$-BPS states twisted by a symmetry element (in the conjugacy class $\rho$) in type II string theory compactified on $K3\times T^2$. 
This was further extended to the generating function of $\tfrac14$-BPS states in the same theory. The generating function in this case was a  genus-two Siegel modular form that we denote by $\Phi^\rho(\mathbf{Z})$\cite{Govindarajan:2009qt}. 

Renewed interest in this moonshine (now called \textit{Mathieu moonshine}) appeared following the work of Eguchi-Ooguri-Tachikawa (EOT) who observed the appearance of the dimensions of irreps of $M_{24}$  in the elliptic genus of $K3$ when expanded in terms of characters of the $\mathcal{N}=4$ superconformal algebra\cite{Eguchi:2010ej}.  The Siegel modular form $\Phi^\rho(\mathbf{Z})$, when it exists, unifies the two Mathieu moonshines, the one related to multiplicative eta products as well as the one related to the elliptic genus\cite{Cheng:2010pq,Govindarajan:2010fu}. 

It has now been established that there is a moonshine that relates conjugacy classes of $M_{24}$ to (the EOT) Jacobi forms of weight zero and index one that we denote by $Z^\rho(\tau,z)$\cite{Gaberdiel:2010ch,Gaberdiel:2010ca,Eguchi:2010fg}.  These Jacobi forms arise
as twistings (also called `twinings') of the elliptic genus of $K3$.  Given a conjugacy class $\rho$ of $M_{24}$,  the associated Jacobi form
expressed in terms of $\mathcal{N}=4$ characters: the massless
character (with zero isospin) $\mathcal{C}(\tau,z)$ and the massive characters $q^{h-\frac18}\,\mathcal{B}(\tau,z)$ (with $h\geq 0$) takes the form\cite{Eguchi:2008gc,Eguchi:2009cq}
\begin{equation}\label{M24char}
Z^\rho(\tau,z) = \alpha^{\,\rho}~ \mathcal{C}(\tau,z) + q^{-\tfrac18}~\Sigma^{\,\rho}(\tau)~ \mathcal{B}(\tau,z)\ ,
\end{equation}
where  $\alpha^{\,\rho}= 1+\chi_{23}(\rho)$ and the character expansion of the function $\Sigma^\rho(\tau)$ is as follows:
\begin{multline}
\Sigma^\rho(\tau)= -2 + [\chi_{45}(\rho)+ \chi_{\overline{45}}(\rho) ]\ q + [\chi_{231}(\rho) +\chi_{\overline{231}}(\rho) ]\ q^2 \\ + [\chi_{770}(\rho)+ \chi_{\overline{770}}(\rho) ]\ q^3  
 + 2 \chi_{2277}(\rho)\ q^4 + 2 \chi_{5796}(\rho)\ q^5 +\cdots   
\end{multline}
where the subscript denotes the dimension of the irrep of $M_{24}$ and $q=e^{2\pi i\tau}$.
An all-orders proof of the existence of such an expansion has been given by Gannon\cite{Gannon:2012ck}. In particular, $\Phi^\rho(\mathbf{Z})$ can be constructed in two ways: an additive lift which uses the eta product $\eta_\rho(\tau)$  as input and a multiplicative lift where
$Z^\rho(\tau,z)$ is the input. However, the additive lift is not known for all conjugacy classes.

In some cases, the square-root of $\Phi^\rho(\mathbf{Z})$ is related to a Borcherds-Kac-Moody (BKM) Lie superalgebra with the additive and multiplicative lifts providing the sum and product side of the Weyl denominator formula.  An attempt at understanding this square-root was done in \cite{Govindarajan:2010cf} where it was argued that there might be a moonshine involving the Mathieu group $M_{12}$ relating its conjugacy classes to BKM Lie superalgebras. In this paper, we revisit that proposal from several viewpoints. Our results may be summarised as follows:
\begin{enumerate}
\item We study a conjecture \ref{weakform} (due to one of us) that implies a moonshine for $M_{12}$ that provides Jacobi forms of weight zero and index 1 for every conjugacy class of $M_{12}$. We find 1427 families of Jacobi forms that have a positive definite character expansion. This result, albeit non-unique, is sufficient to show that all the Fourier coefficients of the EOT Jacobi forms for $M_{24}$ conjugacy classes that reduce to conjugacy classes of $\Mtwo$ (a maximal sub-group of $M_{24}$) are even. This provides an alternate group-theoretic proof of a result due to Creutzig et al.\cite{Creutzig:2012}.
\item In an attempt to obtain a unique solution, we introduce a stronger form  of the conjecture (Proposition \ref{phases}) -- this imposes a  condition that the Jacobi forms transform with suitable phases under an appropriate level. 
For conjugacy classes $4a/4b/8a/8b$, we find \textbf{no} solutions, thereby concluding that there is no moonshine for $M_{12}$.
\item We address the non-existence  of a moonshine for $M_{12}$ by providing two moonshines for $L_2(11)$ that arise as two distinct sub-groups of $M_{12}$. Using this, we construct Siegel modular forms for all conjugacy classes using a product formula that arises naturally as a consequence of $L_2(11)$ moonshine. The modularity of this product is proven in two ways: (i) as an additive lift determined by the eta product, $\eta_{\hat\rho}(\tau)$ -- this works for all but three conjugacy classes ($1^111^1$, $3^4$ and $6^2$), and (ii) as a product of rescaled Borcherds products -- this works for all conjugacy classes.
\item For all conjugacy classes of $L_2(11)$, we show the existence of Borcherds-Kac-Moody (BKM) Lie superalgebras for conjugacy classes of both the $L_2(11)$ by showing that the Siegel modular forms, $\Delta_k^{\hat\rho}(\mathbf{Z})$, arise as their  Weyl-Kac-Borcherds deominator identity. Figure \ref{moonshinesfig} pictorially summarises the moonshines for $L_2(11)$.
\end{enumerate}

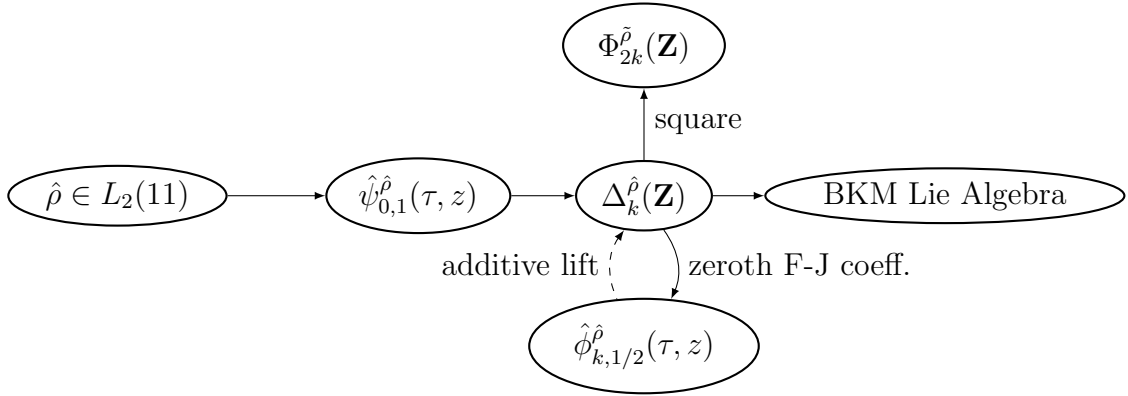
\begin{figure}[hbt]\label{moonshinesfig}
\centering
\begin{tikzpicture}
\draw (0,0)node[ellipse, thick, inner sep=2pt,draw](y1){$\hat\rho\in L_2(11)$};
\draw (4,0)node[ellipse, thick, inner sep=2pt,draw](y2){$\hat\psi_{0,1}^{\hat\rho}(\tau,z)$};
\draw[-latex](y1)--(y2);
\draw (7,0)node[ellipse, thick, inner sep=2pt,draw](y3){$\Delta_k^{\hat\rho}(\mathbf{Z})$};
\draw[-latex](y2)--(y3);
\draw (7,-2)node[ellipse, thick, inner sep=4pt,draw](x3){$\hat{\phi}_{k,1/2}^{\hat\rho}(\tau,z)$};
\draw (7,2)node[ellipse, thick, inner sep=4pt,draw](x4){$\Phi_{2k}^{\tilde\rho}(\mathbf{Z})$};
\draw[latex-](x3) to [bend right]node[right]{zeroth F-J coeff.} (y3);
\draw[latex-, dashed](y3) to [bend right]node[left]{additive lift} (x3);
\draw[latex-](x4)--node[right]{square}(y3);
\draw (11,0)node[ellipse, thick, inner sep=2pt,draw](y4){BKM  Lie Algebra};
\draw[-latex](y3)--(y4);
\end{tikzpicture}
\caption{Moonshines for $L_{2}(11)$: $\hat\rho$ is an $L_{2}(11)_{A/B}$ conjugacy class.}
\end{figure}

The plan of the manuscript is as follows. Following the introductory section, 
in section 2, we describe the conjecture \ref{weakform} for $M_{12}$ moonshine. We find multiple solutions that satisfy all the constraints that are imposed. In proposition \ref{phases}, we show that a stronger form of $M_{12}$ moonshine does not hold. We then show that there are unique solutions for two distinct $L_2(11)$ subgroups of $M_{12}$. In section 3, we construct genus-two Siegel modular forms for all conjugacy classes for both $L_2(11)$ moonshines using a multiplicative lift. In section 4, we show that each of these Siegel modular forms arises as the Weyl-Kac-Borcherds denominator formula for a BKM Lie superalgebra.
We conclude with brief remarks in section 5. Appendix A contains various definitions and details of modular forms, Jacobi forms and Siegel modular forms. Appendix B provides some details of the computations that prove modularity of the multiplicative lift. Finally, Appendix C
contains details of (finite) group theory that are relevant for this paper. \\

\noindent\textbf{Notation} \\
\begin{tabular}{l l }
$\rho$ & Conjugacy class of $M_{24}$ \\

$\tilde\rho$ & Conjugacy class of $M_{12}:2$ \\

$\hat\rho$ & Conjugacy class of $M_{12}$ and $L_2(11)$ \\

$\varrho$ & Weyl vector for a BKM Lie superalgebra \\

$\hat\rho_m$ & Conjugacy class of the $m$-th power of an element of conjugacy class $\hat\rho$ \\

$\eta_\rho(\tau)$ & Eta product for cycle shape $\rho$\\


$Z^\rho(\tau,z)$ & EOT Jacobi form for conjugacy class $\rho$ of $M_{24}$ \\


$\widehat\psi^{\hat\rho}(\tau,z)$ & Jacobi form of weight zero and index one for  conjugacy class $\hat\rho$ of $M_{12}$ \\
 
$\phi^{\hat\rho}_{k,1/2}(\tau,z)$ & Additive seed corresponding to $M_{12}$ conjugacy class $\hat\rho$\\

$\Phi^\rho_k(\mathbf{Z})$ & Siegel modular form of weight $k$ for conjugacy class $\rho$ of $M_{24}$\\

$\Delta^{\hat\rho}_k(\mathbf{Z})$ & Siegel modular form of weight $k$ for conjugacy class $\hat\rho$ of $M_{12}$ and $L_2(11)$\\
\end{tabular}

\section{The $M_{12}$ conjecture}

In some cases, the square-root of the Siegel modular form that unifies the additive and multiplicative moonshines for $M_{24}$ is related to the denominator formula  of a Borcherds-Kac-Moody (BKM) Lie superalgebra. A necessary condition is that the multiplicative seed have even coefficients. A group-theoretic answer to this question comes via 
$\Mtwo$, a maximal subgroup of $M_{24}$ that can be constructed from $M_{12}$ and its outer automorphism. We denote characters and conjugacy classes of $\Mtwo$ by
$\tchi_i$ and $\trho$ throughout this paper. Similarly, we will use a hat for $M_{12}$ characters and conjugacy classes.

$\Mtwo$ is a maximal subgroup of $M_{24}$.  The construction of $\Mtwo$ has its origin in the  order two outer automorphism of $M_{12}$   that is generated by an element that we denote  by  $\varphi$. There are two classes of elements, $(\hat{g},e)$ and $(\hat{g},\varphi)$, where $\hat{g}\in M_{12}$. The composition rule is given by ($\hat{g}_1,\hat{g}_2\in M_{12}$ and $h \in (e,\varphi)$)
\[
(\hat{g}_1,e) \cdot (\hat{g}_2,h) = (\hat{g}_1\cdot \hat{g}_2,h)\quad,\quad(\hat{g}_1,\varphi) \cdot (\hat{g}_2,h) = (\hat{g}_1\cdot \varphi(\hat{g}_2),\varphi\cdot h)
\]
The existence of the decomposition of $Z^\rho(\tau,z)$ in terms of characters of $M_{24}$ as in Eq. \eqref{M24char} immediately implies the following decomposition for all the 21 conjugacy classes of
$\Mtwo$.
\begin{equation}\label{M122char}
Z^{\trho}_{0,1}(\tau,z) = \alpha^{\,\trho}~ \mathcal{C}(\tau,z) + q^{-\tfrac18}~\Sigma^{\,\trho}(\tau)~ \mathcal{B}(\tau,z)\ ,
\end{equation}
where  $\alpha^{\,\trho}= 2+\tchi_2(\trho)+ \tchi_{3}(\trho)$ and the function $\Sigma^\trho(\tau)$ can be expanded in terms of characters of $\Mtwo$ as follows\footnote{We indicate conjugacy classes and other objects related to  $\Mtwo$ with a tilde and a hat for $M_{12}$. In addition, characters of $\Mtwo$ are labelled with the beginning letters of the alphabet while letters beginning from $m$ are used for characters of $M_{12}$.} :
\begin{align}\label{M122Sigma}
\Sigma^\trho(\tau)= -2 +  \sum_{n=1}^\infty \left(\sum_{a=1}^{21} \widetilde{N}_a(n) \tchi_a(\trho) \right) q^n\ .
\end{align}
where the multiplicities $\widetilde{N}_{a}(n)$ are non-negative integers. 

Characters of representations of $\Mtwo$ arise in two ways from characters of $M_{12}$. For  representations of \textit{splitting} type,  a pair of $\Mtwo$ characters are determined by a single $M_{12}$ character and for representations  of \textit{fusion} type, a $\Mtwo$ character is determined by a couple of $M_{12}$ characters.
\begin{center}
\begin{tabular}{c|c|c}
\multirow{2}{*}{ Rep. Type} 
      & \multicolumn{2}{c}{Conjugacy Class} 
                    \\ \cline{2-3} 
 & $(g,e)$ & $(g,\varphi)$ \\ \hline
 Splitting  & $\tilde{\chi}_a = \tilde{\chi}_{a'} =\hat{\chi}_m$ & $\tilde{\chi}_a + \tilde{\chi}_{a'} =0$ \\ 
  Fusion & $\tilde{\chi}_a = \hat{\chi}_m+\hat{\chi}_{m'}$ & $\tilde{\chi}_a =0$ \\[5pt]
\end{tabular}
\end{center}

 Let $\hat{\rho}$ denote a conjugacy class of $M_{12}$ associated with $\hat{g}\in M_{12}$.  Then, the pair of conjugacy classes $(\hat{\rho},\varphi(\hat{\rho}))$ become a conjugacy class $\tilde{\rho}$ associated with the element $(\hat{g},e)$ of 
  $M_{12}\!:\!2$ and hence of $M_{24}$ as well. $12$ of the $21$ conjugacy classes of $\Mtwo$ arise in this fashion. 

  \begin{conj}[Govindarajan\cite{Govindarajan:2010cf}] \label{weakform}
There exists a moonshine for $M_{12}$ that associates  a unique weight zero, index one real Jacobi form $\widehat{\psi}_{0,1}^{\hat\rho}$ to every conjugacy class $\hat{\rho}$ of $M_{12}$ such that
\begin{enumerate}
\item $Z^{\tilde\rho}(\tau,z)= \widehat{\psi}_{0,1}^{\hat\rho}(\tau,z) + \widehat{\psi}_{0,1}^{\varphi(\hat\rho)}(\tau,z)$, where $\tilde\rho$ is the conjugacy class of the element $(g,e)\in\Mtwo$, $g\in M_{12}$ is in the conjugacy class $\hat\rho$
 and $Z^{\tilde\rho}_{0,1}(\tau,z)$ is the EOT Jacobi form that appears in the moonshine for $M_{24}$.
\item The Jacobi form written in terms of $\mathcal{N}=4$ characters \[ \widehat{\psi}_{0,1}^{\hat\rho}(\tau,z)= \hat\alpha^{\hat\rho}~ \mathcal{C}(\tau,z) + q^{-\tfrac18}~\hat{\Sigma}^{\hat\rho}(\tau)~ \mathcal{B}(\tau,z)\ ,\]
where $\hat\alpha^{\hat\rho}$ and   $\hat{\Sigma}^{\hat\rho}(\tau)$ can be expressed  in terms of $M_{12}$ characters. One has
$\hat\alpha^{\hat\rho}= 1+\hat{\chi}_{2}(\hat\rho)$ and 
\begin{equation} \label{M12Sigma}
\hat{\Sigma}^{\hat{\rho}}(\tau) = -1 + \sum_{n=1}^\infty \left(\sum_{m=1}^{15} \hat{N}_m(n) \hat{\chi}_m(\hat\rho) \right) q^n\ .
\end{equation}
with  $\hat{N}_{m}(n)$ being non-negative integers  for all $m\geq 1$ and $n\geq 1$.
 \end{enumerate}
 \end{conj}

\subsection{Implications of  Conjecture \ref{weakform}}

We will later see  that Conjecture \ref{weakform} holds only after one relaxes the uniqueness condition on the Jacobi forms. We will now study its implications.

\begin{prop}\label{easycases}
For $M_{12}$ conjugacy classes $\hat\rho\neq 4a/4b/8a/8b$
   \begin{align*}
     {\widehat\psi}_{0,1}^{\hat\rho}(\tau,z)=\frac{1}{2}\ Z^{\tilde\rho}(\tau,z)\quad,\quad \ .
     \end{align*}
\end{prop}
\noindent \textbf{Remark:} It is easy to verify that $\alpha^{\hat\rho}=\alpha^{\varphi(\hat\rho)}$ for these conjugacy classes.
Thus it suffices to show that $\hat{\Sigma}^{\hat{\rho}} = \hat{\Sigma}^{\varphi(\hat{\rho})}$. A complete list of EOT Jacobi forms that are related to conjugacy classes of $\Mtwo$ is given in Table \ref{EOTJFlist}.
\begin{proof}
For $M_{12}$ conjugacy classes $\hat\rho$ such that $\varphi(\hat\rho)=\hat\rho$, conjecture \ref{weakform} implies that $\widehat{\psi}_{0,1}^{\hat\rho}(\tau,z) = \widehat{\psi}_{0,1}^{\varphi(\hat\rho)}(\tau,z)$ and hence they are given by half the corresponding $\Mtwo$ Jacobi form. This excludes the conjugacy classes $4a/4b/8a/8b/11a/11b$.  The $M_{12}$ characters $\hat{\chi}_4$ and $\hat{\chi}_5$ take complex values for the conjugacy classes $11a/b$.  The reality condition on the Jacobi forms (in Conjecture \ref{weakform})  requires that the multiplicities that appear in Eq. \eqref{M12Sigma} must be such that $$ \hat{N}_4(n)= \hat{N}_5(n) \text{ for all } n\geq 1 \ . $$ Further, it follows that 
$\widehat{\psi}_{0,1}^{11a}(\tau,z) = \widehat{\psi}_{0,1}^{11b}(\tau,z)$.
\end{proof}
\noindent \textbf{Remark:} There is a proposal due to Eguchi and Hikami where they propose a moonshine called `Enriques' moonshine\cite{Eguchi:2013es}. For the four classes that are undetermined, they propose to take one half of the $M_{24}$ Jacobi form as the required Jacobi form. That is not consistent with the decomposition in the expression for $\hat\alpha^{\hat\rho}$ as it implies that $\hat\alpha^{\hat\rho}= 1 + \tfrac12 (\hat{\chi}_2 + {\hat\chi}_3$) -- the coefficients are half integral and do not make sense from group theory. The outer automorphism of $M_{12}$ that we make extensive use of also does not make an appearance in their considerations. We believe that their proposal in not related to our work.
\begin{prop}
For all  conjugacy classes of $\Mtwo$, the Fourier-Jacobi coefficients of the  associated Jacobi form are even i.e.,
\[
Z^\trho(\tau) = 0 \mod 2\ .
\]
\end{prop}
\noindent \textbf{Remark:} This follows from a theorem of Creutzig, H\"ohn and Miezaki\cite{Creutzig:2012} where they show that the $Z^\rho(\tau,z)$ for $M_{24}$ conjugacy classes $(7a/b, 14a/b, 15a/b, 23a/b)$
have some odd coefficients. These are precisely  the conjugacy classes of $M_{24}$ that do not reduce to conjugacy classes of $\Mtwo$. We provide an alternate proof assuming that the $M_{12}$ conjecture holds.
\begin{proof}
It is useful to rewrite the above character decomposition taking into account the outer automorphism $\varphi$ of $M_{12}$. With this in mind, we define 
\begin{align*}
\hat{N}_2^\pm(n)= (\hat{N}_2(n) \pm \hat{N}_3(n))\quad&,\quad \hat{N}_9^\pm(n)= (\hat{N}_9(n) \pm \hat{N}_{10}(n))\ , \\
\hat{\chi}_2^\pm(\hat\rho)=  (\hat\chi_2(\hat\rho)) \pm \hat\chi_3(\hat\rho)))\quad&,\quad\hat\chi_9^\pm(\hat\rho))=   (\hat\chi_9(\hat\rho)) \pm \hat\chi_{10}(\hat\rho))\ .
\end{align*}
Note that $\hat\chi_2^\pm(\varphi(\hat\rho))=\pm \hat\chi_2^\pm (\hat\rho)$ and $\hat\chi_9^\pm(\varphi(\hat\rho))=\pm \hat\chi_9^\pm (\hat\rho)$ by construction.  We can then write the character decomposition as follows
 \begin{multline}
\hat{\Sigma}^{\hat{\rho}} = -1 + \sum_{n=1}^\infty \bigg( \sum_{\substack{m=1\\ m\neq 2,3,9,10}}^{15}  \hat{N}_m(n) \hat{\chi}_m(\hat\rho) +\tfrac12  \hat{N}_2^+(n) \hat\chi^+_2(\hat\rho) + \tfrac12  \hat{N}_9^+(n) \hat\chi^+_9(\hat\rho) \\
+ \tfrac12  \hat{N}_2^-(n) \hat\chi^-_2(\hat\rho) + \tfrac12  \hat{N}_9^-(n) \hat\chi^-_9(\hat\rho) 
\bigg) q^n\ .
\end{multline}

Since the $\mathcal{N}=4$ characters have integer coefficients, it suffices to show that $\alpha^\trho=0\mod 2$ and $\Sigma^\trho=0\mod 2$ for all conjugacy classes.
\begin{itemize}
\item[Part 1:] $\alpha^\trho= 1 + \tchi_2(\trho)+\tchi_3(\trho)=0\mod 2$ for all conjugacy classes as can be explicitly checked.
\item[Part 2:] Consider the  $\Mtwo$ conjugacy classes that arise from elements of type $(g,e)$. For such cases, 
\begin{align}
\Sigma^\trho(\tau) &= \hat{\Sigma}^{\hat{\rho}} + \hat{\Sigma}^{\varphi(\hat{\rho})}\ , \nonumber \\
&=  -2 +  \sum_{n=1}^\infty \bigg(\Big[\!\!\!\!
\sum_{\substack{m=1\\ m\neq 2,3,9,10}}^{15} \!\!\! 
2\hat{N}_m(n) \hat{\chi}_m(\hat\rho) \Big]+ \hat{N}^+_2(n)\hat{\chi}^+_2(\hat\rho)
+ \hat{N}^+_9(n)\hat{\chi}^+_9(\hat\rho) 
\bigg) q^n \ ,\nonumber 
\end{align}
Comparing the above equation with Eq. \eqref{M122Sigma}, we obtain
\begin{equation} \label{splittingrelation}
2\hat{N}_m(n) = \tilde{N}_{a}(n)+\tilde{N}_{a'}(n) \text{ for splitting characters}\ ,
\end{equation}
and for the fusion type,
\[
\hat{N}^+_2(n) = \tilde{N}_3(n) \quad,\quad \hat{N}^+_9(n) = \tilde{N}_{11}(n)\quad,\quad  2\hat{N}_4(n)= 2\hat{N}_5(n)= \tilde{N}_4 (n)\ .
\]
Next, considering the expression for $\Sigma^\trho(\tau)$ above modulo 2, we obtain
\begin{align}
\Sigma^\trho(\tau) 
 = \hat{N}^+_2(n)\hat{\chi}^+_2(\hat\rho)
+ \hat{N}^+_9(n)\hat{\chi}^+_9(\hat\rho) \mod 2\ .
\end{align}
Further, one has $\hat{\chi}^+_2(\hat\rho)=\hat{\chi}^+_9(\hat\rho)=0\mod 2$ as one can explicitly check. Thus, for conjugacy classes for elements of type $(g,e)$, one has $\Sigma^\trho(\tau)=0\mod 2$.
\item[Part 3:] Now consider conjugacy classes of $\Mtwo$ associated with elements of type $(g,\varphi)$. For these conjugacy classes, the character for all fusion representations vanish. Thus,
\begin{align*}
\Sigma^\trho(\tau)&= -2 +  \sum_{n=1}^\infty \left(\sum_{a\in \text{splitting}} \widetilde{N}_a(n) \tchi_a(\trho) \right) q^n  \\
&= -2 +  \sum_{n=1}^\infty \bigg(\sum_{\substack{\text{pairs }\\(a,a')\in \text{splitting}}} (\widetilde{N}_a(n) \tchi_a(\trho) + \widetilde{N}_{a'}(n) \tchi_{a'}(\trho) \Big) q^n \\
&= -2 +  \sum_{n=1}^\infty \left(\sum_{\substack{\text{pairs }\\(a,a')\in \text{splitting}}} (\widetilde{N}_a(n) - \widetilde{N}_{a'}(n)) \tchi_a(\trho)  \right) q^n \ ,
\end{align*}
where we have used the relation $\tchi_a + \tchi_{a'}=0$ for all pairs of splitting representations. The characters for the pairs $(7,8)$ and $(16,17)$ are irrational for these conjugacy classes. For such pairs, rationality (and hence integrality) of the Fourier-Jacobi coefficients of the EOT Jacobi forms implies that $ (\widetilde{N}_a(n) - \widetilde{N}_{a'}(n))=0$ for $(a,a')=(7,8), (16,17)$. For all other pairs, equation \eqref{splittingrelation} implies the weaker condition, $ (\widetilde{N}_a(n) - \widetilde{N}_{a'}(n))=0 \mod 2$. Thus, for conjugacy classes for elements of type $(g,\varphi)$, one has $\Sigma^\trho(\tau)=0\mod 2$.

\end{itemize}
\end{proof}

\subsection{Checking the $M_{12}$ conjecture}

We have seen in Proposition \ref{easycases} that for $M_{12}$ conjugacy classes $\wrho\neq 4a/4b/8a/8b$, one has
\begin{equation}
\wpsi^\wrho(\tau,z) = \frac12\, Z^\rho(\tau,z)\ ,
\end{equation}
where $ Z^\rho(\tau,z)$ is the EOT Jacobi form for the $M_{24}$  conjugacy class $\rho=(\wrho)^2$ i.e., $\rho=1^{2a_1}2^{2a_2}\cdots N^{2a_N}$ if $\wrho=1^{a_1}2^{a_2}\cdots N^{a_N}$. That leaves four undetermined Jacobi forms associated with the classes $4a/4b$ and $8a/8b$. In these cases, we have the relation that relates the sums of Jacobi form of the two $M_{12}$ conjugacy classes to EOT Jacobi forms.
\begin{align}
Z^{4b}(\tau,z) &= \wpsi^{\widehat{4a}}(\tau,z) +  \wpsi^{\widehat{4b}}(\tau,z) \\
Z^{8b}(\tau,z) &= \wpsi^{\widehat{8a}}(\tau,z) +  \wpsi^{\widehat{8b}}(\tau,z) \ .
\end{align}
Thus it suffices to determine $\wpsi^{\widehat{4a}}(\tau,z)$ and $\wpsi^{\widehat{8a}}(\tau,z)$ to obtain Jacobi forms for all conjugacy classes. These two examples have vanishing twisted Witten index as the corresponding cycle shapes have no one-cycles. Thus, for $\wrho=4a$ and $8a$, one has
\begin{equation}
\wpsi^{\wrho}(\tau,z) = \gamma^\wrho(\tau) \   \frac{\vartheta_1(\tau,z)^2}{\eta(\tau)^6}\ ,
\end{equation}
where $\gamma^\wrho(\tau)$ is a weight two modular form of suitable subgroup of $SL(2,\mathbb{Z})$.

\subsubsection*{Constraints}

There are two kinds of constraints that we impose on the weight two modular forms.
\begin{enumerate}
\item \textbf{Non-negativity of coefficients in the character expansion:} We anticipate that all the Jacobi forms associated with the fifteen $M_{12}$-conjugacy classes admit a decomposition in terms of the fifteen $M_{12}$ characters. The $\mathcal{N}=4$ decomposition (as carried out by Eguchi-Hikami\cite{Eguchi:2008gc,Eguchi:2009cq}) implies 
\begin{equation}\label{characterexpand}
\wpsi^\wrho_{0,1}(\tau,z) = \widehat{\alpha}^{\,\wrho}~ \mathcal{C}(\tau,z) + q^{-\tfrac18}~\widehat{\Sigma}^{\,\wrho}(\tau)~ \mathcal{B}(\tau,z)\ ,
\end{equation}
where  $\widehat{\alpha}^{\,\hat{\rho}}= 1+\widehat{\chi}_2(\hat{\rho})$ and
\begin{align}
\widehat{\Sigma}^{\wrho}(\tau)= -1 + \widehat{\chi}_6\ q + [\widehat{\chi}_{8}+\widehat{\chi}_{15} ]\ q^2  + [\widehat{\chi}_{11}+2\ \widehat{\chi}_{13}+2\ \widehat{\chi}_{14}+\widehat{\chi}_{15}) ]\ q^3   + \cdots \nonumber
\end{align}
We have written out the first four terms in the character expansion of $\widehat{\Sigma}^{\wrho}$ as there is no ambiguity arising from the undetermined conjugacy classes. The ambiguity arises from irreps that get exchanged by the outer automorphism of $M_{12}$ -- these correspond to the four characters $\wchi_2\leftrightarrow \wchi_3$ and $\wchi_9\leftrightarrow \wchi_{10}$. In particular, this implies that we know the first four terms in the $q$-series for $\gamma^{\widehat{4a}}$ and $\gamma^{\widehat{8a}}$. One has
\begin{equation}\label{initialterms}
\gamma^{\widehat{4a}}= 1- 4q +4q^2 +4q^3 + \cdots \quad,\quad
\gamma^{\widehat{8a}}= 1- 2q -2q^2 +2q^3+\cdots
\end{equation}
The constraint from group theory is that coefficients in the character expansion of $\widehat{\Sigma}^\wrho$ are  all \textit{non-negative} integers.
\item \textbf{Modularity:} The multiplicative eta products for the two classes $4a/8a$ are modular forms at level $16$ and $64$ respectively. From our experience with determining such examples for the $M_{24}$ Jacobi forms for conjugacy classes with no one-cycles, these do provide a rough guide to determining the levels for our Jacobi forms. For the class $4a$, the first four terms already determined imply that the level must be larger than $8$. 
\end{enumerate}
We find solutions where both conjugacy classes are determined by weight two modular forms at level 32. 

\subsection{The conjugacy classes $4a/8a$}
Our non-unique proposal for the Jacobi forms for the conjugacy classes $4a/8a$ are as follows:
\begin{equation}
\begin{aligned}
\widehat{\psi}^{\widehat{4a}}(\tau,z) &=\Big[ \gamma_{4a}(\tau)+ 4  \alpha(\tau) \Big]\  \frac{\vartheta_1(\tau,z)^2}{\eta(\tau)^6}\ ,\\
\widehat{\psi}^{\widehat{8a}} (\tau,z)&=\Big[ \gamma_{8a}(\tau)-2 \alpha(\tau) \Big]\  \frac{\vartheta_1(\tau,z)^2}{\eta(\tau)^6}\ ,
\end{aligned}
\end{equation}
where $\gamma_{4a}(\tau)$ and $\gamma_{8a}(\tau)$ are the following weight two modular forms of $\Gamma_0(32)$:
\begin{align}
\gamma_{4a}(\tau) &:= -\tfrac{31}{192} E_2^{(2)}(\tau) + \tfrac{37}{64}  E_2^{(4)}(\tau)-\tfrac{77}{48}  E_2^{(8)}(\tau)+\tfrac{35}{16}  E_2^{(16)}(\tau)  - \tfrac{9}4 f(\tau)\nonumber \\ &\qquad\qquad  -\tfrac12 \eta_{4^28^2}(\tau) 
\\
&= 1- 4q +4q^2 +4q^3-4q^4 -20 q^5 +16q^6 + 8q^7 +24q^8 +\cdots   \nonumber \\
\gamma_{8a}(\tau) &:=\tfrac{41}{384} E_2^{(2)}(\tau) - \tfrac{45}{128}  E_2^{(4)}(\tau)+\tfrac{301}{192}  E_2^{(8)}(\tau)
-\tfrac{155}{32}  E_2^{(16)}(\tau)
\nonumber  \\
  &\qquad\qquad  +\tfrac{217}{48}  E_2^{(32)}(\tau)+\tfrac38 f(\tau) + \tfrac12 f(2\tau)-\tfrac{13}4 \eta_{4^28^2}(\tau) \\
  &= 1- 2q -2q^2 +2q^3+2q^4 + 14q^5 -12 q^6 + 4q^7 -32 q^8 + \cdots \ .\nonumber 
 \end{align}
The ambiguity in our solution is  given by a modular form  $\alpha(\tau)$ with integral coefficients and is parametrised by four integers $(d_4,d_5,d_6,d_8)$.
\begin{align*}
\alpha(\tau) = &\phantom{+}d_4 ( -\tfrac{1}{32}  E_2^{(4)}(\tau)+\tfrac{7}{64}  E_2^{(8)}(\tau)-\tfrac{5}{64}  E_2^{(16)}(\tau) ) \\
&+ d_5(-\tfrac{1}{256} E_2^{(2)}(\tau) + \tfrac{1}{256}  E_2^{(4)}(\tau)-\tfrac{1}{16} f(\tau) + \tfrac18\eta_{4^28^2}(\tau)) \\
&+d_6 (-\tfrac{1}{384} E_2^{(2)}(\tau) + \tfrac{3}{256}  E_2^{(4)}(\tau)-\tfrac{7}{768}  E_2^{(8)}(\tau) -\tfrac18 f(2\tau)) \\
&+ d_8( -\tfrac{7}{192}  E_2^{(8)}(\tau)+\tfrac{15}{128}  E_2^{(16)}(\tau) -\tfrac{31}{384}  E_2^{(32)}(\tau)) \\
=& \phantom{+} d_4\ q^4 - d_5\ q^5 + d_6\ q^6 + d_8\ q^8 + \cdots .
\end{align*}
We can see that these four integers would be  determined if we could fix either $\gamma^{\widehat{4a}}$ or $\gamma^{\widehat{8a}}$ to order $q^8$. Non-negativity of the multiplicities defined in Eq. \eqref{M12Sigma} for the first few values of $n$ leads to the following inequalities:
\[
-4\leq d_4 \leq 0\quad,\quad -3\leq (d_5 - 3d_4) \leq 3 \ , \]
\[
-10 \leq d_6 +3(3d_4 -d_5) \leq 8
\quad ,\quad -32 \leq  d_8 + 9 d_6  - 22 d_5 +51 d_4\leq 46 \ . 
\]
There exists no solution with $d_5=-1$ and $d_6=d_8=0$. Thus there is no solution for which $\gamma^{\widehat{4a}}$ is a modular form of  $\Gamma_0(16)$.
To order $q^{128}$, we find 1427 solutions that include $d_4=d_5=d_6=d_8=0$. We have checked that  all these solutions continue to satisfy the positivity constraints to order $q^{512}$ and possibly to all orders\footnote{All 1427 solutions have been listed in the LaTeX source file  after the \texttt{$\backslash$end\{document\}}command. They can be accessed by downloading the source file from arXiv.org.}.

The result can also be understood in terms of the expansion of $\hat\Sigma^{\hat\rho}$ in terms of representations of $M_{12}$. The above solution completely determines the multiplicities
$\hat{N}_m(n)$ defined in Eq. \eqref{M12Sigma} for all representations except $m=9,10$.  In particular, it uniquely fixes $\hat{N}_2(n)$ and $\hat{N}_3(n)$. Table \ref{characterdecomp} lists out the multiplicities of the various $M_{12}$ characters, i.e., $\hat{N}_m(n)$, in the character decomposition for the $M_{12}$ Jacobi form for $n\in[0,32]$ when $d_4=d_5=d_6=d_8=0$.

\subsection{There is no moonshine for $M_{12}$}

The EOT Jacobi forms associated with elements of order $N$ transform as Jacobi form at level $N$ with phases that are (powers of) twelfth-roots of unity\cite{ChengDuncan:2012}. Further, it is known that $H_3(M_{24},\mathbb{Z})=\mathbb{Z}_{12}$.  
The non-uniqueness of our solution for $M_{12}$ moonshine suggests we look for further constraints beyond the ones that we have imposed. It is known that like $M_{24}$, $M_{12}$ has a non-trivial 3-cocycle\cite{DutourSikiric2009}. One has $H_3(M_{12},\mathbb{Z})= \mathbb{Z}_8 \oplus \mathbb{Z}_6$. With this in mind we looked to see if whether suitable powers of the  1427 solutions for $\gamma^{\widehat{4a}}$ and $\gamma^{\widehat{8a}}$  are modular forms of $\Gamma_0(4)$ and $\Gamma_0(8)$ respectively. We did not find any solution. A more detailed analysis  where we go beyond looking at just the 1427 solutions leads to a stronger negative result.

\begin{prop}\label{phases}
 There exist no modular forms $\gamma^{\widehat{4a}}$ and $\gamma^{\widehat{8a}}$, whose initial terms are as given in Eq. \eqref{initialterms}, that transform with phases given by (powers of) eighth-roots of unity at levels 4 and 8 respectively such that the non-negativity conditions on the multiplicities $\hat{N}_2(n)$, $\hat{N}_3(n)$, $\hat{N}_9(n)$  and $\hat{N}_{10}(n)$ defined in Eq. \eqref{M12Sigma} holds for all $n$.
\end{prop}
\begin{proof} Let $\tilde{N}_a(n)$ denote multiplicity of the representation with $\Mtwo$ character $\tilde{\chi}_a$ in the expansion of $\Sigma^{\tilde{\rho}}$ and similarly $\hat{N}_m(n)$ is defined 
for $M_{12}$. Since $\hat{N}_2(n)$, $\hat{N}_3(n)$, $\hat{N}_9(n)$  and $\hat{N}_{10}(n)$ remain unfixed, we write them as follows:
$$ \hat{N}_2(n)=\frac{\tilde{N}_3(n)}2 + c_2 (n),\ \hat{N}_3(n)=\frac{\tilde{N}_3(n)}2 - c_2 (n)\ ,$$
and
$$ \hat{N}_9(n)=\frac{\tilde{N}_{11}(n)}2 + c_9 (n),\ \hat{N}_{10}(n)=\frac{\tilde{N}_{11}(n)}2 - c_9 (n)\ .$$
We can now write the Fourier expansion of weight two modular form for $\gamma^{\widehat{4a}}$ and $\gamma^{\widehat{8a}}$ in terms of multiplicities with $c_2(n)$ and $c_9(n)$ as unknowns. From the positive definiteness of the multiplicities we have following constraints.
$$|c_2(n)| \le  \frac{\tilde{N}_3(n)}2 \text{  and  } |c_9(n)| \le  \frac{\tilde{N}_{11}(n)}2\ .$$
Eq. \eqref{initialterms} determines $c_2(n)$ and $c_9(n)$ for $n\leq 4$ with the others remaining unfixed. 

We look for them to be modular forms with unknown phases that are powers of an eighth-root of unity. We take the eighth powers of  $\gamma^{\widehat{4a}}$ and $\gamma^{\widehat{8a}}$ -- these are expected to be  modular forms of weight 16 at level 4 (of dimension 9) and 8 (of dimension 17) respectively. 
 The remaining 16 coefficients can be expressed in terms of $c_2(n)$ ($n\in[5,16]$) and $c_9(n)$ ($n\in[5,8]$) by matching up to order $q^8$ for $\gamma^{\widehat{4a}}$ and order $q^{16}$ for $\gamma^{\widehat{8a}}$. Modularity determines all other $c_2(n)$ and $c_9(n)$. 
Let us focus on $c_9(9)$ which can be expressed in terms of the 16 unknowns. We obtain
 \begin{multline*}
 c_9(9) = -1078695 + 2359 c_2(6) - 396 c_2(7) + 31 c_2(8) - c_2(9) - 
    65930 c_9(4) \\ - 420 c_9(4)^2  - 1830 c_9(5) + 28 c_9(4) c_9(5) + 
   2359 c_9(6) - 396 c_9(7) + 31 c_9(8)\ .
 \end{multline*}
 Using the constraints from positivity on the coefficients appearing on the right hand side of
 the above equation, we get
 $$-1254891 \le c_9(9) \le -905523\ .$$
Further, positive definiteness of $\hat N_9(9)$ implies $-75\le c_9(9) \le 75$.
 Clearly the two constraints are not compatible with each other. Thus, there exists no solution that is compatible with the positive definiteness of $\hat N_9(9)$.
\end{proof}

\subsection{Moonshines for $L_2(11)$ }

We have seen that there is  no moonshine for $M_{12}$. With this in mind, 
we  look for subgroups of $M_{12}$ for which the characters $\hat\chi_2^-(\hat\rho)$ and $\hat\chi_9^-(\hat\rho)$ vanish on restriction to the sub-group. There are two such
sub-groups, both isomorphic to $L_2(11)$. The first is a maximal subgroup of $M_{12}$ and the second is a maximal subgroup of $M_{11}\subset M_{12}$. As sub-groups of $M_{12}$, these two groups are \textit{not} conjugate to each other and thus lead to distinct moonshines\cite{Conway1971}.

\subsubsection{$L_{2}(11)$}
$L_2(11)$ is Artin's notation for the finite simple  group $PSL(2,\mathbb{F}_{11})=SL(2,\mathbb{F}_{11})/\mathbb{F}_{11}^{\times}$, where $\mathbb{F}_{11}$ is the prime field of integers modulo $11$. It has a natural action on the projective line, $PL(11)$,
via projective linear transformations:
\[
x \rightarrow \frac{a x +b }{c x +d} \quad, \quad x\in PL(11)\ .
\]
The projective line $PL(11)$ consists of 12 points whose inhomogeneous coordinates are given by the set $\Omega =(0,1,2,3,\ldots, 9, X=10, \infty)$. This provides a 12-dimensional permutation representation of $L_2(11)$. In this representation, $L_2(11)$ is generated as $\langle  \alpha, \beta, \gamma\rangle=:L_2(11)_A$, where
\begin{equation}
\alpha: x \rightarrow x+1 \quad,\quad \beta: x \rightarrow 3 \cdot x \quad, \quad \gamma: x \rightarrow -1/x \quad.
\end{equation}
Explicitly, one has
\begin{align*}
\alpha&=(\infty)(0,1,2,3,4,5,6,7,8,9,X) \\
\beta &=(\infty)(0)(1,3,9,5,4) (2,6,7,X,8) \\
\gamma &= (\infty,0) (1,X),(2,5),(3,7) (4,8) (6,9)
\end{align*}
One has $\alpha^{11}=\beta^5=\gamma^2=1$. 
The eight conjugacy classes of $L_2(11)$ are given by the following cycle shapes  in $L_2(11)_A$:
\[
\begin{array}{c|rrrccrcc}
\rho &1a & 2a &3a& 5a &5b &6a &11a & 11b \\[3pt] \hline
\text{cycle shape} & 1^{12} & 2^6 & 3^4 & 1^25^2 & 1^2 5^2 & 6^2 & 1^1 11^1 & 1^111^1\phantom{\Big|}\\
\text{element} & 1 & \gamma & \alpha \gamma &  \beta &\beta^{-1} &  \alpha\gamma\beta & \alpha & \alpha^{-1} \\
\end{array} 
\]
Let $\delta$ represent the permutation (with cycle shape $1^42^4$) acting on $PL(11)$:
\[
\delta=(\infty) (0)(1)(2,X)(3,4)(5,9)(6,7)(8)\ .
\]
A second construction of $L_2(11)$, that we call $L_2(11)_B$, is generated by $\langle \alpha, \beta, \delta\rangle$.
All three generators fix $\infty$ and  thus $L_2(11)_B$ permutes points in $\Omega \backslash \infty$.  The cycle shapes for 
the conjugacy classes for  $L_2(11)_B$ are
\[
\begin{array}{c|cccccccc}
\rho &1a & 2a &3a& 5a &5b &6a &11a & 11b \\[3pt] \hline
\text{cycle shape} & 1^{12} & 1^4 2^4 & 1^33^3 & 1^25^2 & 1^2 5^2 & 1^12^13^16^1 & 1^1 11^1 & 1^111^1\phantom{\Big|} \\
\text{element} & 1 & \delta & \alpha \delta &  \beta &\beta^{-1} &  \alpha \delta \beta & \alpha & \alpha^{-1} \\
\end{array}
\]

The important observation here is that both $L_2(11)_A$ and $L_2(11)_B$ do not have any elements of order 4 and 8. Thus the conjugacy classes $4a/4b$ and $8a/8b$ do not reduce to conjugacy classes of these sub-groups. The conjugacy class $10a$ of $M_{12}$, for which we do know the Jacobi form, also does not appear.

Below we provide first few terms that appear in the character expansion  which is  the analog of Eq. \eqref{M12Sigma} for the two $L_{2}(11)$ subgroups. For $L_2(11)_A\subset M_{12}$, one has\footnote{For the two equations that follow, the characters that appear are those for $L_2(11)$. }
\begin{multline}
\Sigma = -\chi _1+ \left(\chi _1+2 \chi _5+\chi _7+\chi _8\right)q + \left(\chi _2+\chi _3+5 \chi _4+2 \chi
   _5+5 \chi _6+4 \chi _7+4 \chi _8\right) q^2 \\ + \left(\chi _1+8 \chi _2+8 \chi _3+9 \chi _4+12 \chi
   _5+13 \chi _6+14 \chi _7+14 \chi _8\right)q^3\\ + \left(2 \chi _1+15 \chi _2+15 \chi _3+39 \chi _4+32
   \chi _5+37 \chi _6+42 \chi _7+42 \chi _8\right) q^4
   +O\left(q^5\right)
\end{multline}
For $L_1(11)_B\subset M_{11}\subset M_{12}$
\begin{multline}
\Sigma = -\chi _1+ \left(\chi _4+\chi _6+\chi _7+\chi _8\right) q+ \left(\chi _1+3 \chi _2+3 \chi _3+2 \chi
   _4+4 \chi _5+4 \chi _6+4 \chi _7+4 \chi _8\right)q^2\\ 
   + \left(\chi _1+4 \chi _2+4 \chi _3+15 \chi
   _4+10 \chi _5+13 \chi _6+14 \chi _7+14 \chi _8\right) q^3 \\
   +\left(4 \chi _1+19 \chi _2+19 \chi _3+31
   \chi _4+38 \chi _5+35 \chi _6+42 \chi _7+42 \chi _8\right)q^4
   +O\left(q^5\right)
\end{multline}

\section{Siegel Modular Forms for $L_2(11)_{A}$ and $L_2(11)_B$}

\begin{table}
\centering
\begin{tabular}{c|c|c|c|c|c}
$M_{12}$ Conj. Class & Cycle Shape $\hat\rho$ & $N_{\hat\rho}$ & $p_{\hat\rho}$ &$k_{\hat\rho}$ &  $\wchi_{\hat\rho}(d)$ \\[3pt] \hline
1a & $1^{12}$ & $1$ & 1&$6$ &\phantom{$\left(\tfrac{-1}{d}\right)$}\\
2a & $2^6$ & $2$ & $2$ &$3$ & $\left(\tfrac{-1}{d}\right)$\\
2b & $1^4 2^4$ & $2$ & $1$ & $2$  \\
3a & $1^3 3^3$ & $3$ &$3$ & $3$ & $\left(\tfrac{-3}{d}\right)$\\
3b & $3^4$ & $3$ & $3$ & $2$ \\
4a & $2^24^2$ & $4$ & $2$ & $2$ \\
4b & $1^42^2 4^4 / 2^2 4^2$ & $4$ & $2$& $3$ &$\left(\tfrac{-1}{d}\right)$\\
5a & $1^2 5^2 $  & $5$ & $1$ & $2$\\
6a & $6^2$ & $6$ & $6$ &  $1$&$\left(\tfrac{-1}{d}\right)$\\
6b & $1^12^13^1 6^1$ & $6$ & $1$ & $2$ \\
8a & $4^18^1$ & $8$ & $4$ & $1$& $\left(\tfrac{-2}{d}\right)$\\
8b & $1^2 2^14^18^2/ 4^1 8^1$ &  $8$ & $4$ & $2$\\
10a & $2^1 10^1 $  & $10$ & $2$ & $1$&$\left(\tfrac{-20}{d}\right)$ \\[3pt]
11a/b & $1^1 11^1 $  & $11$ & $1$ & $1$&$\left(\tfrac{-11}{d}\right)$ \\[3pt] \hline
\end{tabular}\caption{$M_{12}$ conjugacy classes and the corresponding cycle shapes. The associated eta products are modular forms of $\Gamma_0(2N_{\hat\rho}\,p_{\hat\rho},2)$ with weight $k_{\hat\rho}$ and Dirichlet character $\chi$ with $\Gamma_1(2N_{\hat\rho}\,p_{\hat\rho},2)\subset \text{ker}(\chi)$.  The values are given in columns 3-5 }\label{TableEtaProducts}
\end{table}

The construction of Borcherds-Kac-Moody Lie algebras is intimately connected to modular forms that appear as the Weyl denominator formula for the BKM Lie algebra. In this section, we shall pursue this approach by constructing genus-two Siegel modular forms. First, we show that the two distinct $L_2(11)$ moonshines naturally lead to a product formula given by Eq. \eqref{productformulamain}. Modularity of this formula is not manifest in the construction and we prove this in two ways -- (i) by constructing the sum side as an additive lift and (ii) by showing that the product formula is equivalent to Borcherds products. The second method always works while the additive lift works in most cases.

\subsection{The multiplicative lift}

The connection with $L_2(11)$ moonshine leads to a Siegel modular form that is given by the following formula (on repeating arguments given in \cite{Govindarajan:2011em}):
\begin{align}\label{multiplicativelift}
\Delta_k^{\hat\rho}(\mathbf{Z}) = s^{1/2}\ {\hat\phi}^{\hat\rho}_{k,1/2}(\tau,z)\times \exp\Big(\sum_{m=1}^\infty s^m {\hat\psi}^{\hat\rho}\big|_0 T(m)(\tau,z)\Big)\ ,
\end{align}
where the twisted  Hecke operator (first defined in \cite{Govindarajan:2011em}) is given by
\[
s^m\ \hat{\psi}^{\hat\rho}\big|_0 T(m)(\tau,z) =s^m \  \tfrac1{m} \sum_{\substack{ad=m\\ b \text{ mod }d}}  \hat{\psi}^{\hat\rho_a}\left(\tfrac{a\tau+b}{d},az\right)\
\]
where $\rho_a$ is the conjugacy class of the $a$-th power of an element in the conjugacy class $\rho$. Further, ${\hat\phi}^{\hat\rho}_{k,1/2}(\tau,z)$ is defined as follows:
\begin{equation}\label{additiveseed}
\hat{\phi}^{\hat\rho}_{k,1/2}(\tau,z) = \frac{\vartheta_1(\tau,z)}{\eta(\tau)^3} \times \eta_{\hat\rho}(\tau)\ ,
\end{equation}
where $\eta_{\hat\rho}(\tau)$ is an eta product, $k_{\hat\rho}$ its weight as given in Table \ref{TableEtaProducts} and $k=(k_{\hat\rho}-1)$.

Again, as in \cite{Govindarajan:2011em} with $L_2(11)$ playing the role of $M_{24}$, we will show that Eq. \eqref{multiplicativelift} implies a product formula for $\Delta_k^{\hat\rho}(\mathbf{Z})$. Let $g$ be an element of order $N$ and $\hat{\rho}_a$ denote the conjugacy class of the element $g^a$. Further, set $\hat\rho_0=1^{12}$. Define the Fourier coefficients, $c^a(n,\ell)$, of the Jacobi form $\hat{\psi}^{\hat{\rho}_a}_{0,1}(\tau,z)$ as follows
\begin{equation}
\hat{\psi}^{\hat{\rho}_a}_{0,1}(\tau,z) = \sum_{n=0}^\infty \sum_{\ell\in\mathbb{Z}} c^{a}(n,\ell)\ q^nr^\ell\ .
\end{equation}
Then, define $f_\alpha(n,\ell)$ via the discrete Fourier transform ($\omega_N=\exp(2\pi i/N)$):
\[
c^{a}(n,\ell) = \sum_{\alpha=0}^{N-1} \omega_N^{\alpha a}\ f_\alpha(n,\ell)\ .
\]
One then can rewrite the formula for the multiplicative lift as follows:
\begin{equation}\label{productformulamain}
\Delta^{\hat{\rho}}_k(\mathbf{Z}) = s^{\frac12}\ \hat\phi^{\hat{\rho}}_{k,1/2}(\tau,z) \times \prod_{\alpha=0}^{N-1} \prod_{m=1}^\infty \prod_{n=0}^\infty \!\!\prod_{\substack{\ell\in \BZ\\[2pt] 4nm -\ell^2 \geq 0}}   \Big(1-\omega_N^\alpha\, q^{n} r^{\ell}s^{m}\Big)^{f_{\alpha}(nm,\ell)
}\ , 
\end{equation}
For the cases when  the order of $g\in L_2(11)$ is prime (i.e., $N=2,3,5,11$), one has the conjugacy class of $g^a$ for $a\neq 0 \textrm{ mod }N$ is the same as that of $g$. Thus, one has $\rho_a=\rho$ for $a\neq 0\textrm{ mod } N$.  For these cases, on using the product representation for the theta and eta functions that appear in $\hat\phi^{\hat{\rho}}_{k,1/2}(\tau,z)$, the above formula simplifies to
\begin{equation} \label{primeproduct}
\Delta^{\hat{\rho}}_k(\mathbf{Z}) = q^{\frac12}r^{\frac12}s^{\frac12} 
\prod_{(n,\ell,m)> 0}  \Big(1- q^{n} r^{\ell}s^{m}\Big)^{c^{1}(nm,\ell)}
\Big(1- q^{nN} r^{\ell N}s^{mN}\Big)^{\frac{c^{0}(nm,\ell)-c^1(nm,\ell)}N }
\ , 
\end{equation}
where $(n,\ell,m)> 0$ implies $n>0$, or $n=0$ and $m>0$, or $n=m=0$ and $\ell<0$. \\

For $N=6$, the product formula takes the form
\begin{multline}\label{sixproduct}
\Delta^{\hat{\rho}}_k(\mathbf{Z}) = q^{\frac12}r^{\frac12}s^{\frac12} 
\prod_{(n,\ell,m)> 0}  \Big(1- q^{n} r^{\ell}s^{m}\Big)^{c^1(nm,\ell)}
 \Big(1- (q^{n} r^{\ell }s^{m})^6\Big)^{f_{1}(nm,\ell)} \\
\Big(1- (q^{n} r^{\ell }s^{m})^2\Big)^{\frac12(c^{2}(nm,\ell)-c^1(nm,\ell))
}
  \Big(1- (q^{n} r^{\ell}s^{m})^3\Big)^{\frac13(c^{3}(nm,\ell)-c^1(nm,\ell)) }
\ ,
\end{multline}
with $f_1(n,\ell)=\frac16 (c^0(n,\ell)+c^1(n,\ell)-c^2(n,\ell)-c^3(n,\ell))$.\\

\noindent \textbf{Remarks:} All the terms in the product formula that appear with $m=0$ arise from the product representation of $\hat\phi^{\hat{\rho}}_{k,1/2}(\tau,z)$. Further, the formula is manifestly symmetric under the exchange $q\leftrightarrow s$ and odd under $r\rightarrow r^{-1}$.

\subsection{Modularity of the multiplicative lift}

The product formulae that we obtained  starting from Eq. \eqref{multiplicativelift} is not standard in the context of Siegel modular forms. Thus, we need to establish the modular properties of the product formula. We establish modularity in the next couple of sub-sections.

\subsubsection{Modularity by comparing with the additive lift}

 The additive `seed' for the Siegel modular form is a Jacobi form of weight $k=(k_{\hat\rho}-1)$ and index $1/2$ and is defined in Eq. \eqref{additiveseed}.
For $k>0$, the additive lift is given by  
\begin{equation}\label{additivelift}
\mathcal{A}\Big(\hat\phi^{\hat\rho}_{k,1/2}\Big)(\mathbf{Z}):= \sum_{m=1}^\infty s^{(2m-1)/2} \ \hat\phi^{\hat\rho}\big|_{k,1/2} T_-^M(2m-1)(\tau,z)\ ,
\end{equation}
where $T_-^M(m)$ is  the Hecke operator defined by Clery-Gritsenko\cite{Gritsenko:2008}. Let $\phi(\tau,z)$ be a Jacobi form of weight $k$ of $\Gamma_0(M)$ with character $\chi$ and index which is integral or half-integral. Then
\[
\phi^\rho\big|_{k} T_-^M(m)(\tau,z) = m^{k-1} \sum_{\substack{ad=m\\ (a,Mq)=1\\ b \text{ mod }d}} d^{-k}\ \chi(a)\ \phi^\rho\left(\tfrac{a\tau+q b}{d},az\right)\ .
\]
where $q$ is chosen such that $\Gamma_1(Mq,q)\subset\text{ker}(\chi)$.\footnote{The group $\Gamma_1(Mq,q)$ is defined as follows.
\[\left\{\begin{pmatrix} a & b \\ c & d \end{pmatrix}\in SL(2,\mathbb{Z}) ~\Big|~ c=0 \text{ mod }Mq, \ b=0 \text{ mod }q, a=1\text{ mod }Mq, d=1 \text{ mod }Mq\right\}\ .\]} For all the cases of interest, one has $q=2$ and $M=p_{\hat\rho}N_{\hat\rho}$ as given in Table \ref{TableEtaProducts}. The additive lift $\mathcal{A}\Big(\hat\phi^{\hat\rho}_{k,1/2}\Big)(\mathbf{Z})$ is a genus-two Siegel modular form of a level $N$ subgroup of $Sp(2,\mathbb{Z})$ with weight $k$.

A necessary condition for the compatibility of the  additive lift   with the multiplicative lift is:
\[
 \Big[\tfrac{\theta_1(\tau,z)}{\eta(\tau)^3}\ \eta_{\hat\rho}(\tau)\Big]\Big|_{k} T_-^M(3)(\tau,z)\ \stackrel{?}= \hat\psi^{\hat\rho}(\tau,z)   \Big[\tfrac{\theta_1(\tau,z)}{\eta(\tau)^3}\ \eta_{\hat\rho}(\tau)\Big]
\]
This is the coefficient of $s^{3/2}$ in both the lifts i.e., the ones given in Eq. \eqref{additivelift} and Eq. \eqref{multiplicativelift}. This condition holds for all the cycle shapes appearing in $L_2(11)_{A/B}$ except for the following three cycle shapes where we observe that:
\begin{align}
\frac{T_3 {\hat\phi}^{3^4}}{\hat\phi^{3^4}} -\hat\psi^{3^4} &= \frac{\theta_1(\tau,z)^2}{\eta(\tau)^6} \Big[ 9\ \eta_{1^33^{-2}9^3}(\tau)    \Big] \nonumber\ ,  \\
\frac{T_3 \hat\phi^{1^111^1}}{\hat\phi^{1^111^1}} -\hat\psi^{1^1 11^1} &= \frac{\theta_1(\tau,z)^2}{\eta(\tau)^6} \Big[ \tfrac{11}3\ \eta_{1^211^2}(\tau)    \Big]\ , \\
\frac{T_3 \hat\phi^{6^2}}{\hat\phi^{6^2}} -\hat\psi^{6^2} &= \frac{\theta_1(\tau,z)^2}{\eta(\tau)^6} \Big[ \tfrac{2}3 \eta_{6^4}(\tau) +\tfrac13 \eta_{1^33^{-3}9^3}(\tau) + 2    \eta_{1^33^{-3}9^3}(2\tau) + \tfrac83 \eta_{1^33^{-3}9^3}(4\tau)  \Big] \nonumber\ ,
\end{align}
where $T_3\hat\phi^{\hat\rho}$ is short for $\hat \phi^{\hat\rho}|_k T_-^M(3)$. The last two examples potentially correspond to Siegel modular forms with weight $k=0$ and we have 
na\"ively applied a formula which assumes $k>0$. So we need to prove modularity in another way for these three examples.

\subsubsection{Modularity by comparing with a Borcherds formula}

We begin with a theorem due to Clery-Gritsenko (see also \cite{GritsenkoNikulinI,GritsenkoNikulinII,Aoki:2005}) that leads to a Borcherds product formula for a meromorphic Siegel modular form starting from a nearly holomorphic Jacobi form of weight zero and index $t$ specialised to the case for the case $t=1$.  

  \begin{theorem}[Clery-Gritsenko\cite{Gritsenko:2008}] \label{CGproduct} Let $\psi$ be a nearly holomorphic Jacobi form of weight $0$ and index $1$ of $\Gamma_0(N)$. Assume that for all cusps of $\Gamma_0(N)$ one has $\frac{h_e}{N_e} c_{f/e}(n,\ell)\in \mathbb{Z}$ if $4n -\ell^2\leq 0$. Then the product
  \[
  B_\psi(\mathbf{Z}) = q^A r^B s^C \prod_{f/e\in \mathcal{P}} \prod_{\substack{n,\ell,m\in \mathbb{Z}\\ (n,\ell,m)>0}} \Big(1-(q^nr^\ell s^{m})^{N_e}\Big)^{\frac{h_e}{N_e} c_{f/e}(nm,\ell)}\ ,
  \]
  with 
  \[
  A=\frac1{24} \sum_{\substack{f/e\in \mathcal{P}\\ \ell \in \mathbb{Z}}} h_e\, c_{f/e}(0,\ell) ,\ 
  B=\frac1{2} \sum_{\substack{f/e\in \mathcal{P}\\ \ell \in \mathbb{Z}_{>0}}} \ell h_e\, c_{f/e}(0,\ell),\ 
   C=\frac1{4} \sum_{\substack{f/e\in \mathcal{P}\\ \ell \in \mathbb{Z}}} \ell^2 h_e\, c_{f/e}(0,\ell)\ ,
  \]
  defines a meromorphic Siegel modular form of weight 
  \[
  k=  \frac1{2} \sum_{\substack{f/e\in \mathcal{P}\\ \ell \in \mathbb{Z}}} \frac{h_e}{N_e} c_{f/e}(0,0)
  \]
  with respect to $\Gamma_1(N)^+$ possibly with character. The character is determined by the zeroth Fourier-Jacobi coefficient of $B_\psi(\mathbf{Z})$ which is a Jacobi form of weight $k$ and index $C$ of the Jacobi subgroup of $\Gamma_1(N)^+$.
  \end{theorem}
  \noindent \textbf{Remark:} As discussed by Clery and Gritsenko, the poles and zeros of $B_\psi$ lie on rational quadratic divisors defined by the Fourier coefficients $ c_{f/e}(n,\ell)$ for $4n-\ell^2\leq0$. The condition $\frac{h_e}{N_e} c_{f/e}(n,\ell)\in \mathbb{Z}$ ensures that one has only poles or zeros at these divisors. In our case, multiple cusps contribute to the same term in the product formula and hence we relax the condition.
for all cusps that have identical values of $(N_e,h_e)$, we require that the sum of $\frac{h_e}{N_e} c_{f/e}(n,\ell)$ (with $4n-\ell^2\leq0$) for all such cusps be integral. This suffices to ensure that one has only zeros or poles at all divisors.
  
  Since only the coefficients $c_{f/e}(n,\ell)$,  with $n\in \mathbb{Z}$ appear in the product formula, we define the projection(also defined in \cite{Raum:2012}), $\pi_{FE}$, as follows
  \begin{equation}
  \pi_{FE} \left(\phi |M_{f/e}\right) (\tau,z) := \frac1{h_e} \sum_{b=0}^{h_e-1} \phi |M_{f/e} (\tau +b,z)\ ,
  \end{equation}
  where $M_{f/e} = \left(\begin{smallmatrix} f & * \\ e & *\end{smallmatrix}\right)\in SL(2,\mathbb{Z} )$ maps the cusp at $i\infty$ to $f/e$. It is the Fourier coefficients of the projected Jacobi form that appears in the product formula. 

Following Raum\cite{Raum:2012}, we look to prove modularity of the product given in Eq. \eqref{productformulamain} by considering products of rescaled Borcherds products.  Our considerations not only extend his results but also provide a systematic method of obtaining the precise rescaled Borcherds products that are needed. For all conjugacy classes of $L_2(11)_B$, we find that the product formula is equivalent to a single Borcherds formula. More generally, we find that the following holds.   
  \begin{prop}\label{modularprop}
For $g\in L_2(11)_{A/B}$, let $\hat\rho_m=[g^m]$, $\hat\rho=\hat\rho_1$, $\psi^{\hat\rho_m}=\hat\psi_{0,1}^{\hat\rho_m}(\tau,z)$ and  $\Delta_k^{\hat\rho}(\mathbf{Z})$ be the modular form defined by the multiplicative lift given in Eq. \eqref{productformulamain}. Then
\begin{equation}\label{modularproduct}
\left(\Delta_k^{\hat\rho}(\mathbf{Z})\right)^{p_{\hat\rho}}  = \prod_{m|p} \left(B_{\frac{p}m\psi^{\hat\rho_m}}(m\,\mathbf{Z})\right) \ ,
\end{equation}
where $p_{\hat\rho}$ is the length of the shortest cycle in the cycle shape for the conjugacy class $\hat\rho$ and $B_\psi(\mathbf{Z})$ is the Siegel modular form given by Theorem \ref{CGproduct}. 
  \end{prop}

  \begin{proof}
  
We deal with case when $p_{\hat\rho}=1$ before considering $p_{\hat\rho}>1$. \\

\noindent $\boxed{\mathbf{p_{\hat\rho}=1}}$\\

 This occurs for all cases when $g\in L_2(11)_B$ and the conjugacy classes of order 1, 5, and 11 in $L_2(11)_A$. In all these cases, there is precisely one term in the product appearing on the right hand side of Eq. \eqref{modularproduct}.

 Let $N$ be the order of $g$.  The Jacobi form $\psi$ is a modular form of $\Gamma_0(N)$ with weight $0$ and index $1$. Considering the case when $N$ is prime, there are only two cusps, one at $i\infty$ (which is $\Gamma_0(N)$ equivalent to $1/N$) and another at $0/1$.  To prove the equality, we need to show that
\begin{align}
\pi_{FE} (\psi) &=\psi_{0,1}^{\hat\rho} \ ,\\
\pi_{FE} (\psi|S) &= \frac1N (\psi_{0,1}^{1^{12}} -\psi_{0,1}^{\hat\rho})\ . \label{check2}
\end{align}
The first equation holds trivially since $\psi=\psi_{0,1}^{\hat\rho} $ has only integral powers of $q$ in its Fourier-Jacobi expansion. The second part follows from a calculation.
\[
\pi_{FE}(\psi|S)(\tau,z) = \frac{1}{N+1} \phi_{0,1} (\tau,z) + \pi_{FE}(\alpha^{(N)}|S)(\tau) \phi_{-2,1}(\tau,z)\ ,
\]
where $\alpha^{(N)}(\tau)=\frac{N}{N+1} E_2^{(N)}(\tau)$ for $N=2,3,5$. Computing $\pi_{FE}(\alpha^{(N)}|S)(\tau)$, we obtain
\[
\pi_{FE}(\alpha^{(N)}|S)(\tau) = -\tfrac{1}{N(N+1)} \sum_{b=0}^{N-1} E_2^{(N)}(\tfrac{\tau+b}N)= -\tfrac{1}{(N+1)} E_2^{(N)}\big|_2U_N = -\tfrac{1}{(N+1)} E_2^{(N)}(\tau)\ ,
\]
where $U_N$ is the Hecke operator for $\Gamma_0(N)$ (defined by Atkin and Lehner\cite{Atkin:1970}) and  $E_2^{(N)}$ is its eigenform with eigenvalue $+1$.
Thus, we get
\begin{align*}
\pi_{FE}(\psi|S)(\tau,z) &= \frac{1}{N+1} \phi_{0,1} (\tau,z) -\frac1N \alpha^{(N)}(\tau)\phi_{-2,1}(\tau,z)\\
&= \frac{1}{N} \phi_{0,1} (\tau,z) -\frac1N\psi_{0,1}^{\hat\rho} (\tau,z)\\
&= \frac1N \left(\psi_{0,1}^{1^{12}}-\psi_{0,1}^{\hat\rho}\right)(\tau,z)\ ,
\end{align*}
which establishes Eq. \eqref{check2} for prime $N=2,3,5$.
A similar computation holds for $N=11$ for which
$\alpha^{(11)}(\tau)=\frac{11}{6} E_2{(11)}(\tau)-\frac{22}5 \eta_{1^211^2}(\tau)$. One can show that $\pi_{FE}(\alpha^{(11)}|S)(\tau)=-(1/11)\alpha^{(11)}(\tau)$. thus Eq. \eqref{check2} holds for $N=11$ as well.

Next, considering the case of $N=6$, where there are additional cusps at $1/3$ (with width $2$) and $1/2$ (with width $3$). One needs to verify that the last three conditions hold as the first condition holds trivially.
\begin{equation}\label{six-identities}
\begin{split}
\pi_{FE}(\psi) &= \hat\psi_{0,1}^{\hat\rho}(\tau,z)\\
\pi_{FE}(\psi|S) &= \frac16(\hat\psi_{0,1}^{1^{12}}(\tau,z)+\hat\psi_{0,1}^{\hat\rho}(\tau,z)-\hat\psi_{0,1}^{\hat\rho_2}(\tau,z)-\hat\psi_{0,1}^{\hat\rho_3}(\tau,z))\\
\pi_{FE}(\psi|M_{1/3}) &= \frac12(\hat\psi_{0,1}^{\hat\rho_2}(\tau,z)-\hat\psi_{0,1}^{\hat\rho}(\tau,z))\\
\pi_{FE}(\psi|M_{1/2}) &= \frac13(\hat\psi_{0,1}^{\hat\rho_3}(\tau,z)-\hat\psi_{0,1}^{\hat\rho}(\tau,z))
\end{split}
\end{equation}
In the above equations $\hat\rho=1^12^13^16^1$, $\hat\rho_2=1^33^3$ and $\hat\rho_3=1^42^4$. $\psi=\hat\psi_{0,1}^{\hat\rho}(\tau,z)= \frac12 Z^{1^22^23^26^2}(\tau,z)$, is a Jacobi form of $\Gamma_0(6)$. Then,
\begin{align*}
\psi|_{0,1} S (\tau,z) &= \frac16\phi_{0,1} (\tau,z) +\frac1{12} \left( E_2^{(2)}(\tau/2)+ 2  E_2^{(3)}(\tau/3) -5  E_2^{(6)}(\tau/6)\right)\phi_{-2,1}(\tau,z)\ ,\\
\psi|_{0,1} M_{1/2} (\tau,z)   &=  \frac16\phi_{0,1} (\tau,z) +\frac1{12}\left(E_2^{(2)}(\tau) -8 E_2^{(3)}\left(\tfrac{2\tau+1}3\right) +5 E_2^{6}\left(\tfrac{\tau+2}3\right) \right)\phi_{-2,1}(\tau,z)\ ,\\
\psi|_{0,1} M_{1/3} (\tau,z) &=  \frac16\phi_{0,1}(\tau,z)  + \frac1{12}\left(2 E_2^{(3)}(\tau) - 9 E_2^{(2)}\left(\tfrac{3\tau-1}3\right) +5 E_2^{6}\left(\tfrac{\tau-1}2\right)\right)\phi_{-2,1}(\tau,z)\ .
\end{align*}
The projections $\pi_{FE}$ of the Eisenstein series appearing in the above equations are given by (with $b\in\mathbb{Z}$)
\begin{align*}
\pi_{FE}\left(E_2^{(2)}(\tfrac{\tau+b}2)\right) &= E_2^{(2)}(\tau)\ ,\\
\pi_{FE}\left(E_2^{(3)}(\tfrac{\tau+b}3)\right) &= E_2^{(3)}(\tau)\ ,\\
\pi_{FE}\left(5E_2^{(6)}(\tfrac{\tau+b}6)\right) &= -5E_2^{(6)}(\tau)+6 E_2^{(3)}(\tau)+4 E^{(2)}_2(\tau)\ ,
\end{align*}
which implies Eq. \eqref{six-identities}.\\

\noindent $\boxed{\mathbf{p_{\hat\rho}>1}}$\\

There are three conjugacy classes of $L_2(11)_A$, $2^6$, $3^4$ and $6^2$, with $p=2,3,6$ respectively. First, consider $B_{\psi^{\hat\rho}}$ with $\psi^{\hat\rho}$ a Jacobi form of $\Gamma_0(N^2)$ where $N$ is the order of the group element. For all three conjugacy classes, cusps whose width does not divide the order of the group element, do not contribute as the Fourier expansion has no integral powers of $q$ and thus they vanish under the projection $\pi_{FE}$. The details of the computation are given in Appendix \ref{appendixcomp}.
\begin{itemize}
\item[$\mathbf{2^6}$:] We need to consider the contribution from the cusps at $i\infty$ and $1/2$ (with $h_e=1$ $N_e=1$). One has
\[
\pi_{FE}(\psi|M_{1/2})(\tau,z) =-\hat\psi_{0,1}^{\hat\rho}(\tau,z)\ .
\]
We need to consider the square of $\Delta^{2^6}_2(\mathbf{Z})$ so that $\frac{h_e}{N_e} c_{f/e}(n,\ell)\in \mathbb{Z}$ if $4n -\ell^2\leq 0$. 
\[
B_{2\,\psi^{2^{6}}}(\mathbf{Z})  =  
\prod_{(n,\ell,m)> 0}  \Big(1- q^{n} r^{\ell}s^{m}\Big)^{2c^{1}(nm,\ell)}
\Big(1- q^{2n} r^{2\ell}s^{2m}\Big)^{- c^1(nm,\ell) }\ ,
\]
which is a meromorphic Siegel modular form at level 4 with weight $-1$.
This Siegel modular form does not account for the following terms in Eq. \eqref{primeproduct}:
\[
q^{\frac12}r^{\frac12}s^{\frac12}\ \prod_{(n,\ell,m)> 0} (1- q^{2n} r^{2\ell }s^{2m})^{\frac{c^{0}(nm,\ell)}2}\ .
\]
The square of this term is
$B_{\psi^{1^{12}}}(2\mathbf{Z})$. We  obtain
\[
(\Delta^{2^6}_2(\mathbf{Z}))^2 = B_{2\,\psi^{2^{6}}}(\mathbf{Z})\ B_{\psi^{1^{12}}}(2\mathbf{Z}) \ .
\]
Thus, $(\Delta^{2^6}_2(\mathbf{Z}))^2$ is a Siegel modular form of weight $4$ at level $4$. The transformation of the Siegel modular form is the one induced from  
$(\phi_{2,1/2}^{2^6}(\tau,z))^2$. 
\item[$\mathbf{3^4}$:] We need to consider the contribution from the cusps at $i\infty$, $1/3$ (with $h_e=1$ $N_e=3$) and $2/3$ ($h_e=1$ $N_e=3$). The contributions from the cusps at $1/3$ and $2/3$ have cube roots of unity. However,  the contributions of the two cusps add in the product formula to give integral coefficients.  One finds
\[
\pi_{FE}(\psi|M_{1/3}+\psi|M_{2/3})(\tau,z) =-\hat\psi_{0,1}^{\hat\rho}(\tau,z)\ ,
\]
leading to 
\[
B_{3\,\psi^{3^{4}}}(\mathbf{Z})  =
\prod_{(n,\ell,m)> 0}  \Big(1- q^{n} r^{\ell}s^{m}\Big)^{3c^{1}(nm,\ell)}
\Big(1- q^{3n} r^{3\ell}s^{3m}\Big)^{- c^1(nm,\ell) }\ ,
\]
which is a Siegel modular form at level 9 with weight $-2/3$.
This Siegel modular form does not account for the following terms in Eq. \eqref{primeproduct}:
\[
q^{\frac12}r^{\frac12}s^{\frac12}\ (1- q^{3n} r^{3\ell }s^{3m})^{\frac{c^{0}(nm,\ell)}3}\ .
\]
 The cube of this term is
$B_{\psi^{1^{12}}}(3\mathbf{Z})$. Thus we obtain
\[
(\Delta^{3^4}_1(\mathbf{Z}))^3 = B_{3\,\psi^{3^{4}}}(\mathbf{Z})\ B_{\psi^{1^{12}}}(3\mathbf{Z})\ .
\]
Thus, $(\Delta^{3^4}_1(\mathbf{Z}))^3$ is Siegel modular form of weight $3$ at level $9$.
\item[$\mathbf{6^2}$:] We need to consider the contribution from the cusps at $i\infty$, ($1/6, 5/6)$, $(1/12, 5/12)$ and $1/18$. Like in the case of $3^4$, we need to pair up cusps with identical values of $h_e$ and $N_e$ to get integral coefficients. One obtains
\begin{align*}
\pi_{FE}(\psi|M_{1/6}+\psi|M_{5/6}) &=\hat\psi_{0,1}^{\hat\rho}(\tau,z) \\
\pi_{FE}(\psi|M_{1/12}+\psi|M_{5/12}) &=-\hat\psi_{0,1}^{\hat\rho}(\tau,z) \\
\pi_{FE}(\psi|M_{1/18}) &=-\hat\psi_{0,1}^{\hat\rho}(\tau,z)
\end{align*}
\end{itemize}
\begin{multline}
B_{6\, \psi^{6^2}}(\mathbf{Z})  = 
\prod_{(n,\ell,m)> 0}  \Big(1- q^{n} r^{\ell}s^{m}\Big)^{c^1(nm,\ell)}
 \Big(1- (q^{n} r^{\ell }s^{m})^6\Big)^{c^{1}(nm,\ell)} \\
\Big(1- (q^{n} r^{\ell }s^{m})^2\Big)^{-3\, c^1(nm,\ell)}
  \Big(1- (q^{n} r^{\ell}s^{m})^3\Big)^{-2\, c^1(nm,\ell) }
\ ,
\end{multline}
which is a Siegel modular form of weight $-2$ at level 36.
This does not give all the terms that appear in the product form given in Eq. \eqref{sixproduct}. The missing terms can be accounted for by additional terms leading to 
\[
(\Delta^{6^2}_0(\mathbf{Z}))^6 = B_{6\,\psi^{6^2}}(\mathbf{Z})\ B_{2\,\psi^{2^{6}}}(3\mathbf{Z}) B_{3\,\psi^{3^{4}}}(2\mathbf{Z})\ B_{\psi^{1^{12}}}(6\mathbf{Z}) \ .
\]
$(\Delta^{6^2}_0(\mathbf{Z}))^6$ is a meromorphic Siegel modular form of weight zero at level $36$.
  \end{proof}
  
  \begin{prop}\label{characterprop}
  The modular properties of the Borcherds product formula for  $(\Delta_k^{\hat\rho}(\mathbf{Z}))^{2p_{\hat\rho}}$ is determined  by $(\hat\phi_{k,1/2}^{\hat\rho}(\tau,z))^{2p_{\hat\rho}}$. In particular, they are meromorphic Siegel modular forms at level $N_{\hat\rho}$, where $N_{\hat\rho}$ the order of the group element. 
  \end{prop}
  \begin{proof} It is easy to check that $(\hat\phi_{k,1/2}^{\hat\rho}(\tau,z))^{2p_{\hat\rho}}$ is the zeroth Fourier-Jacobi coefficient of $(\Delta_k^{\hat\rho}(\mathbf{Z}))^{2p_{\hat\rho}}$. 
  From the results of Cheng and Duncan\cite{Cheng:2011}, we know that   under $\Gamma_0(N_{\hat\rho})$ transformations,  $(\hat\phi_k^{\hat\rho}(\tau,z))^2$, transforms with the following character
  \[
  \chi(\gamma) = \exp\left(\frac{2\pi i cd}{p_{\hat\rho}\,N_{\hat\rho}}\right) \quad\text{for} \quad\gamma =\begin{pmatrix} a & b \\ c & d \end{pmatrix} \in \Gamma_0(N_{\hat\rho})\ . 
  \]
  The character of the $p_{\hat\rho}$-th power (of the Siegel modular form) is clearly trivial for  all $\gamma\in\Gamma_0(N_{\hat\rho})$. 
 
  \end{proof}
  
\section{BKM Lie algebras for $L_2(11)_{A}$ and $L_2(11)_B$}
  
  We have constructed Siegel modular forms, $\Delta_k^{\hat\rho}(\mathbf{Z})$ for all conjugacy classes of $L_2(11)_{A/B}$. In this section, we will establish that these Siegel modular forms appear as the Weyl-Kac-Borcherds denominator formula for
BKM Lie superalgebras whose identical real simple roots, $(\delta_1,\delta_2,\delta_3)$ have the following rank three hyperbolic Cartan matrix
\[
A^{(1)} =\begin{pmatrix} \phantom{-}2 & -2 & -2 \\ -2 &  \phantom{-}2 & -2 \\ -2 & -2 &  \phantom{-}2 \end{pmatrix}\ .
\]
Each distinct conjugacy class leads to an inequivalent automorphic correction (in the sense of Gritsenko-Nikulin) to the Lorentzian Kac-Moody Lie algebra associated 
to the above Cartan matrix. The imaginary simple roots that appear  depend on the conjugacy class.

Let $w_i$ ($i=1,2,3$) denote  reflections by the three real simple roots and $W=\langle w_1,w_2,w_2\rangle$ be the Weyl group generated by the real simple roots.
The Weyl group acts on the future light-cone $V^+\subset \mathbb{R}^{2,1}=\oplus \mathbb{R}\,\delta_i$. A fundamental domain (under the action of $W$) is a polyhedron $\mathcal{M}$, bounded by three walls. Let  $\mathcal{H}_{\delta_i}:=(x\in V^+~|~ (x,\delta_i)=0)$ be the wall associated with $\delta_i$ and define the space $\mathcal{H}^+_{\delta_i}:=(x\in V^+~|~ (x,\delta_i)\leq 0)$. These walls bound the polyhedron $\mathcal{M}=\cap_{i=1}^3 \mathcal{H}^+_{\delta_i}$. The dihedral group $D_6$  is generated by the involution $\delta_1\leftrightarrow \delta_3$ and the cyclic permutation of the three real simple roots. Under the action of the dihedral group, $\mathcal{M}$ gets mapped to itself.

Let $Q=\oplus_i \mathbb{Z}\, \delta_i$, denote the root lattice and $Q_+=\oplus_i \mathbb{Z}_+\, \delta_i$. The Weyl vector $\varrho=\frac12(\delta_1+\delta_2+\delta_3)$ satisfes $(\varrho,\delta_i)=-1$ for $i=1,2,3$ and is invariant under the dihedral group $D_6$.
We make the formal identifications
\[
e^{-\varrho} \sim  q^{1/2} r^{1/2} s^{1/2}\quad,\quad
e^{-\delta_1} \sim q\, r \quad,\quad
e^{-\delta_2} \sim  r^{-1} \quad,\quad
e^{-\delta_3} \sim s\, r \quad.
\]
Let $\alpha[n,\ell,m]= n \delta_1 + (n+m-\ell) \delta_2 + m \delta_3$. Then, 
one has $e^{-\alpha[n,\ell,m]}\sim q^nr^\ell s^m$ and the norm $(\alpha[n,\ell,m],\alpha[n,\ell,m])= 2\ell^2 - 8nm$.
 
\subsection{Properties of $\Delta_k^{\hat\rho}(\mathbf{Z})$}

We have seen that $\Delta_k^{\hat\rho}(\mathbf{Z})$ is a Siegel modular form of weight $k=(k_{\hat\rho}-1)$ at level $N_{\hat\rho}$ with character. Let $v(M)$ denote this character for  $M$ in the level $N_{\hat\rho}$ subgroup of $Sp(2,\mathbb{Z})$.  Here $N_{\hat\rho}$ is the order of an element of $L_2(11)$ in the conjugacy class $\hat\rho$. 

\begin{enumerate}
\item It is symmetric under the exchange  $\tau \leftrightarrow \tau'$. The corresponding $Sp(2,\mathbb{Z})$ element is called $V$ in appendix \ref{appendixmodforms}. Thus, $v(V)=+1$.
\item It is anti-symmetric under $z\rightarrow -z$. The corresponding $Sp(2,\mathbb{Z})$ element is given by 
\[
\delta=\begin{pmatrix}
1 & 0 & 0 & 0 \\
 0 & -1 & 0 & 0 \\
 0 & 0 & 1 & 0 \\
 0 & 0 & 0 & -1 
\end{pmatrix}\ .
\]
Thus, $v(\delta)=-1$.
\item Under the Heisenberg group, the character is 
\[
v_H([\lambda,\mu;\kappa]_H) = (-1)^{\lambda+\mu +\lambda\mu +\kappa} \ ,
\]
\item Under transformations, $\widetilde{\gamma}\in Sp(2,\mathbb{Z})$ induced by $\gamma\in PSL(2,\mathbb{Z})$, one has
\[
v_{\hat\rho}(\widetilde{\gamma}) := \chi_{\hat\rho}(\gamma)\ ,
\]
where $\chi_{\hat\rho}(\gamma)$ is the character of the eta product $\eta_{\hat\rho}(\tau)$. In particular, there are two ways of understanding this character. The first one, is that the eta product is a modular form of $\Gamma_0(2N_{\hat\rho}\, p_{\hat\rho},2)$ possibly with Dirichlet character as given in Table \ref{TableEtaProducts}. The second one is that the $2p_{\hat\rho}$-th  power of the eta product  is a modular form of $\Gamma_0(N_{\hat\rho})$. We don't give the a precise formula for the character as we don't need it for our considerations. Also, Proposition \ref{characterprop} gives the character for the square of $\Delta_k^{\hat\rho}(\mathbf{Z})$.
\item The Siegel modular form admits the following Fourier expansion: 
\begin{align}
\Delta_k^{\hat\rho}(\mathbf{Z}) 
&= \sum_{\substack{n,\ell,m \equiv 1 \text{ mod }2\\ 4nm-\ell^2>0 \\ n,m > 0}} f(nm,\ell) \ q^{n/2} r^{\ell/2} s^{m/2} \ . \label{second} 
\end{align}
The other conditions $n,\ell,m \equiv 1 \text{ mod }2$ and $n,m > 0$ easily 
follow from the multiplicative lift. For $\hat\rho\neq 1^111^1, 3^4, 6^2$, the condition $(4nm -\ell^2)>0$ follows directly from the additive lift  given in Eq. \eqref{additivelift}.  This is a condition, $\mathcal{D}>0$,  on the discriminant for terms that appear in the $\frac{m}2$-th Fourier-Jacobi coefficient of $\Delta_k^{\hat\rho}(\mathbf{Z})$.

For $\hat\rho=1^111^1, 3^4, 6^2$, where the additive lift is not available, we conjecture that this condition is true as well. One can explicitly verify that it holds for $m=1,3$ as we do below.
\begin{itemize}
\item[$m=1$]
The coefficient of $s^{1/2}$ is the Jacobi form $\hat\phi^{\hat\rho}_{k,1/2} (\tau,z)$.  This Jacobi form has non-vanishing Fourier coefficients  about the cusp at $i\infty$ when the discriminant $\mathcal{D}>0$ which is equivalent to $4n-\ell^2>0$. In particular, the family of terms (that arise from $\theta_1(\tau,z)$ up to an overall pre-factor) with $q^{1/2} r^{1/2} s^{1/2} q^{y(y+1)/2} r^{y}$ ($y\in\mathbb{Z}$) has discriminant $\mathcal{D}=3/4$. More generally, one can show that $\mathcal{D}\geq 3/4$.
\item[$m=3$] From the multiplicative lift, we see that the coefficient of $s^{3/2}$ is 
 \[ \hat\phi^{\hat\rho}_{k,1/2} (\tau,z)\, \hat\psi^{\hat\rho}_{0,1} (\tau,z)\ . \] The Fourier expansion of the multiplicative seed $\hat\psi^{\hat\rho}_{0,1} (\tau,z)$ has terms with negative discriminant, $\mathcal{D}=-1$. These are of the from $q^{n}r^\ell s^m$ with  $n=x(x+1),\ m=1, \ \ell = 2x + 1$ for $x\in \mathbb{Z}$ with $C_{i\infty}(\mathcal{D}=-1,\ell)=+1$ for all conjugacy classes $\hat\rho$. Combining this family of terms with $q^{1/2} r^{1/2} s^{1/2} q^{y(y+1)/2} r^{y}$ ($y\in\mathbb{Z}$) coming from $\hat\phi^{\hat\rho}_{k,1/2} (\tau,z)$, we obtain the term
 \[
 [q^{1/2} r^{1/2} s^{1/2} q^{y(y+1)/2} r^{y}] \times [q^{x(x+1)} r^{2x+1} s] \ ,
 \]
 which has discriminant $\mathcal{D}=(3/4) + 2(x-y)^2 >0$. These are, potentially, the only terms that might have had a negative discriminant. Other terms from $\hat\phi^{\hat\rho}_{k,1/2} (\tau,z)$ come  with higher powers of $q$ that only increase the value of the discriminant. Thus, one has $\mathcal{D}\geq 3/4$.
\end{itemize} 
For $m>3$, we now provide a heuristic argument as to why we expect $\mathcal{D}>0$. All terms with negative discriminant come from the expansion of exponential in Eq. \ref{multiplicativelift}. Further, they arise from the action of the Hecke operator on $\hat\psi^{\hat\rho}_{0,1} (\tau,z)$.  Since terms with discriminant $\mathcal{D}=-1$ that appear in $\hat\psi^{\hat\rho}_{0,1} (\tau,z)$ are the same for all conjugacy classes, we can split, $\hat\psi^{\hat\rho}_{0,1} (\tau,z)$ as
\[
\hat\psi^{\hat\rho}_{0,1} (\tau,z) = \psi_{\mathcal{D}=-1}(\tau,z) + \psi^{\hat\rho}_{\mathcal{D}\geq 0}(\tau,z)\ ,
\]
where the first term which contains \textit{all} terms with negative discriminant. Consider the product (for $y\in\mathbb{Z}$)
\[
 [q^{1/2} r^{1/2} s^{1/2} q^{y(y+1)/2} r^{y}] \times  \exp\Big(\sum_{m=1}^\infty s^m {\hat\psi}_{\mathcal{D}=-1}\big|_0 T(m)(\tau,z)\Big)\ .
\]
Expanding the above, we claim that all terms have discriminant $\mathcal{D}>0$. The dependence on the conjugacy class arises solely from the terms that are not accounted above. Such terms all have $\mathcal{D}\geq 0$ and contain only non-negative powers of $q$ and $s$. Thus, this ensures that all terms that appear in Eq. \eqref{second} must have $\mathcal{D}\geq 0$.\footnote{Let $T_1=q^{n_1}r^{\ell_1}s^{m_1}$ and $T_2=q^{n_2}r^{\ell_2}s^{m_2}$ be such that $n_1,n_2,m_1,m_2\geq0$ and $\mathcal{D}_1>0$, $\mathcal{D}_2\geq 0$. Then, the discriminant, $\mathcal{D}$, of the product $T_1T_2$ is also positive definite, i.e., $\mathcal{D}>0$.} Thus, we expect the conclusion to hold for the three conjugacy classes for which we do not have an additive lift.
\end{enumerate}

\subsection{Establishing the Weyl-Kac-Borcherds denominator formula}

The Weyl-Kac-Borcherds denominator formula takes the form for all the conjugacy classes
\begin{subequations}
\begin{align}
\Delta_k^{\hat\rho}(\mathbf{Z}) 
&= q^{1/2} r^{1/2} s^{1/2} \  \prod_{(n,\ell,m)>0} \Big(1 - q^n r^\ell s^m )^{c(nm,\ell)} \label{first}\\
&= \sum_{w\in W} \det(w) \left(e^{-w(\varrho)}-\sum_{a\in Q\cap \mathcal{M}}m(\alpha)\ e^{-w(\varrho+a)} \right) \label{third}\ 
\end{align}
\end{subequations}
\begin{enumerate}
\item 
The first line, Eq. \eqref{first}, follows  from the multiplicative lift. In particular, Eqns. \eqref{primeproduct} and \eqref{sixproduct} are precisely of this form. Proposition \ref{modularprop}  implies that $c(nm,\ell)\in \mathbb{Z}$.   $c(nm,\ell)$ is the multiplicity of the positive root $\alpha[n.\ell,m]$ and negative values of $c(nm,\ell)$ corresponds to fermionic roots. The condition $(n,\ell,m)>0$ determines all the positive roots.
\item 
Eq. \eqref{third} is the sum side of the denominator formula.
\begin{enumerate}
\item The expansion of all Siegel modular forms  has the following terms 
\[
\begin{split}
\Delta_k^{\hat\rho}(\mathbf{Z}) &=e^{-\varrho} \Big(1- e^{-w_1(\varrho)+\varrho}- e^{-w_2(\varrho)+\varrho}- e^{-w_3(\varrho)+\varrho}+ \cdots\Big) \\
&= q^{1/2} r^{1/2} s^{1/2} \Big( 1 - qr -r^{-1} - sr + \cdots \Big)\ ,
\end{split}
\]
where $e^{\varrho}$ ($\varrho$ is the Weyl vector) is identified with $q^{1/2} r^{1/2} s^{1/2} $ and the three terms shown are the terms corresponding to the real simple roots.
\item  The Siegel modular form is invariant 
the cyclic  $\mathbb{Z}_3$ symmetry that permutes the three real simple roots. It is  generated by 
 the $Sp(2,\mathbb{Z})$ transformation $\tilde{\gamma}=\delta \cdot V\cdot [1,0;0]_H $.
The character $v(\tilde{\gamma})$  is given by 
\[
v(\tilde{\gamma}) = v(\delta)\times v(V)\times v ([1,0;0]_H) = -1 \times 1 \times -1 =+1 \ .
\]
The symmetry of the Siegel modular form under the involution $\tau\leftrightarrow \tau'$ makes it invariant under the dihedral symmetry.
The antisymmetry under $\delta: z\rightarrow -z $ is equivalent to the Weyl reflection $w_2$. When combined with the dihedral symmetry, we see that
\begin{equation}\label{allreflections}
\Delta_k^{\hat\rho}(w_i \cdot \mathbf{Z})=-\Delta_k^{\hat\rho}(\mathbf{Z})\ .
\end{equation}
for $i=1,2,3$.  This implies that for any $w\in W$, one has
\begin{equation}\label{Weylproperty}
\Delta_k^{\hat\rho}(w\cdot \mathbf{Z})=\det(w) \ \Delta_k^{\hat\rho}(\mathbf{Z})\ .
\end{equation}
where $\det(w)=+1$ (resp. $-1$) if $w$ is generated by a combination of even (resp. odd) elementary reflections. 
\item In the following we repeat arguments due to Gritsenko and Nikulin\cite[Theorem 2.3]{Nikulin:1995}  as they are applicable here as well. We can rewrite Eq. \eqref{second} as 
\begin{align}
\Delta_k^{\hat\rho}(\mathbf{Z}) 
&= \sum_{\substack{n,\ell,m \equiv 1 \text{ mod }2\\ ||\alpha[\frac{n}2,\frac{\ell}2,\frac{m}2] ||^2<0\\ n,m > 0}} f(nm,\ell) \ e^{-\alpha[\frac{n}2,\frac{\ell}2,\frac{m}2]} \ . \label{secondp}
\end{align}
 Let $\alpha[\frac{n}2,\frac{\ell}2,\frac{m}2] = \varrho + a[n',\ell',m']$. Then $(n',m')\in \mathbb{Z}_+^2$ and $\ell'\in \mathbb{Z}$.
Due to property \eqref{Weylproperty}, Eq. \eqref{secondp} can be rewritten as follows:
\begin{equation}
\Delta_k^{\hat\rho}(\mathbf{Z}) =\sum_{w\in W} \det(w)\ w\left(\sum_{\varrho+a \in \mathcal{M}\cap \frac12 Q} -m(a)\  e^{-\varrho -a}\right) \label{thirdp}
\end{equation}
where $e^{-a} \sim q^{n'} r^{\ell'} s^{m'}$ and $m(a)=-f(n'+\frac12,m'+\frac12,\ell'+\frac12)$ and $(\varrho+a)$ lies in the Weyl chamber $\mathcal{M}$  i.e., $(\varrho+a,\delta_i)\leq 0$ for 
$i=1,2,3$. 
However, $(\varrho+a,\delta_i)= 0$ does not happen as it needs $(n,m,\ell)=(0,0,0)$ for which $f(0,0,0)=0$. Thus, the stronger condition $(\varrho+a,\delta_i)<0$ 
holds. This  implies that 
$(a,\delta_i)< 1$ for $i=1,2,3$.
 Integrality of $(n',\ell',m')$ implies that $(a,\delta_i)$ must be integral and thus the stronger condition given below holds.
\begin{equation}
(a,\delta_i)\leq 0 \text{ for }i=1,2,3\ .
\end{equation}
Thus $a\in Q\cap \mathcal{M}$. The  equality for all three values of $i$ occurs only when $a=0$ for which one has $m(0)=-1$.  When $a\neq 0$, the above condition implies $(a,a)\leq0$ -- in other words these correspond to imaginary simple roots. Eq. \eqref{third} follows from Eq. \ref{thirdp} after separating the $a=0$ term from the terns with $a\neq0$.
\item The multiplicity of imaginary simple roots with $(a,a)=0$ can be determined by observing that the primitive roots of this type are: $a_0:=a[1,0,0]=(\delta_1+\delta_2)$, $a[0,0,1]=(\delta_2+\delta_3)$ and $a[1,2,1]=(\delta_1+\delta_3)$. These are permuted by the dihedral group and so it suffices to consider $a_0$ and integral multiples of it. The multiplicity of these imaginary simple roots are determined by 
 the zeroth Fourier-Jacobi coefficient, $\phi^{\hat\rho}_k (\tau,z)$. Identifying the $\theta_1(\tau,z)$ with the denominator formula for the $\widehat{sl(2)}$ generated by $\delta_1$ and $\delta_2$, we see that remaining eta products determine the multiplicities by the formula:
\begin{equation}\label{imagrootsmultiplicity}
q^{-3/8}\ \frac{\eta_{\hat\rho}(\tau)}{\eta(\tau)^3}=1- \sum_{m} m(t a_0)\ q^t\ .
\end{equation}
Thus the Siegel modular forms $\Delta_k^{\hat\rho}(\mathbf{Z})$ indeed provide the denominator formula for a family of Borcherds-Kac-Moody super Lie algebras.
\end{enumerate}
\end{enumerate}

\section{Concluding Remarks}

The main positive result of our paper is to show the existence of BKM Lie superalgebras as the end-point of a sequence of moonshines for $L_2(11)_{A/B}$ that sees a beautiful interplay involving multiplicative eta products, EOT Jacobi forms, and Siegel modular forms. It is known that on extending considerations to include CHL orbifolds preserving $\mathcal{N}=4$ supersymmetry leads to other classes of BKM Lie superalgebras whose real simple roots differ from the examples considered here\cite{Govindarajan:2010fu}. The orbifolding group for the CHL orbifolds in the type IIA picture arise from symplectic automorphisms of $K3$. These are known to be sub-groups of $M_{23}$ from the work of Mukai\cite{Mukai:1988}. The conjugacy classes of elements of this group are determined by its order with the order $\leq 8$. Looking at elements of $\Mtwo$ that are also elements of $M_{23}$ picks out conjugacy classes of $\Mtwo$ with atleast 1 one-cycle and  more than four cycles in its cycle shape. This rules out conjugacy classes of $L_2(11)_A$ with orders $2,3$ and $6$ which have no one-cycles and hence are related to $\Mtwo$ conjugacy classes with no one-cycles. The order $11$ elements of $\Mtwo$ with cycle shape $1^211^2$ have only 4 cycles. Thus, the simple CHL $\mathbb{Z}_N$ orbifolds occur for $N\leq 8$.
 In our forthcoming paper\cite{Sutapa2018b}, we revisit these  considerations and  provide evidence for a new type of BKM Lie algebras that arise for the CHL $\mathbb{Z}_5$ and $\mathbb{Z}_6$ orbifolds. The cycle shapes associated with the conjugacy classes of $L_2(11)_A$ with orders $2,3$ and $6$ appear in considering cases involving generalized moonshine associated with commuting pairs of elements. The Jacobi forms that appear in the multiplicative lift here are those that arise in the context of umbral moonshine\cite{Cheng:2012tq}.

The squares of the Siegel modular forms that we construct, i.e., $\Delta_k^{\hat\rho}(\mathbf{Z})$, are Siegel modular forms that appear in the context of $M_{24}$-moonshine. The construction using products of rescaled Borcherds products for the cases when the order of the group element is prime or powers of a prime number  connects up to the work of Raum\cite{Raum:2012}. Other than the conjugacy classes $2^{12}$ and $3^8$ of $M_{24}$,  our results agree when a comparison is possible. For $3^8$, this is due to an incorrect assignment of level during programming and our results correct his. For the order 6 conjugacy classes (i.e., $1^22^23^26^2$ and $6^4$) of $M_{24}$, our results extend Raum's computations. The $M_{12}$ conjugacy class, $2^110^1$, is one where the Jacobi form is one half of an EOT Jacobi form. In this instance, we obtain the following product formula:
\begin{align}
(\Delta_0^{2^1 10^1} (\mathbf{Z}))^{10}= B_{10\psi^{2^1 10^1}} (\mathbf{Z}) B_{5\psi^{1^2 5^2}} (2\mathbf{Z}) B_{\psi^{1^{12}}} (10\mathbf{Z}) \ .
\end{align}
This is different from the formula given in Proposition \ref{modularproduct} where the Siegel modular form, $\Delta_k^{\hat\rho}(\mathbf{Z})$, was raised to the power of the smallest cycle shape which is two in the current example. 
This also provides another example where the na\"ive additive lift fails to match the product formula. One has
\begin{align*}
\frac{T_3 \phi^{2^110^1}}{\phi^{2^110^1}} -\hat\psi^{2^1 10^1} &= \frac{\theta_1(\tau,z)^2}{\eta(\tau)^6} \Big[ \tfrac{20}3\ \eta_{2^210^2}(\tau)    \Big] \ .
\end{align*}
We anticipate that there is a BKM Lie superalgebras associated with this Siegel modular form as well.
Clearly, it appears that we should be able construct Siegel modular forms for all conjugacy classes of $M_{24}$ that don't appear in Raum's list or the ones that we considered here. We do not pursue this here and hope to report it elsewhere.

\medskip

\noindent \textbf{Acknowledgments:}  We thank Aniket Joshi for discussions and for verifying the mod $2$ behaviour of the Jacobi Forms as a part of his undergraduate thesis work under the supervision of one of us (SG).  This work was partially supported by a grant from the Simons Foundation which enabled SG to visit the Aspen Center for Physics during the summer of 2017. SG also thanks the Department of Science and Technology for support through grant EMR/2016/001997 and SS is supported by  NET-JRF fellowship under UGC fellowship scheme.

\appendix

\section{Modular Forms}\label{appendixmodforms}

A modular form, of weight $k$ and character $\chi$, is a function $f:\mathbb{H}\rightarrow \mathbb{C}$ such that for $\gamma=\begin{pmatrix} a& b \\ c& d\end{pmatrix}\in PSL(2,\mathbb{Z})$, one has
\begin{equation}
f|_k \gamma(\tau) = \chi(\gamma)\ f(\tau)\ ,
\end{equation}
where
\[
f|_k \gamma(\tau) := (c\tau + d)^{-k}\ f(\gamma\cdot \tau)\ ,
\]
and $\gamma\cdot \tau = \frac{a\tau+b}{c\tau+d}$. The level $N$ sub-group $\Gamma_0(N)\subseteq PSL(2,\mathbb{Z})$ is given by restricting to $\gamma$ with $c=0\text{ mod} N$.

\subsection{Weight two modular forms}

The Eisenstein series at level $N>1$ and weight $2$ is defined as follows:
$$
E^{(N)}_2(\tau):=\tfrac{12i}{\pi(N-1)}\partial_\tau \big[\ln \eta(\tau) -\ln \eta(N\tau)\big] = 1 + \tfrac{24}{N-1} q + \cdots \ .
$$
Note that $\frac{N-1}{24} E^{(N)}_2(\tau)$ has integral coefficients except for the constant term.
Let $f(\tau)$ denote the following weight two modular form of $\Gamma_0(16)$:
\begin{equation}
f(\tau) := \frac14 \left(\eta_{4^88^{-4}}(\tau)- \eta_{1^42^2 4^{-2}}(\tau)\right)= q -4q^3 + 6q^{5} -8 q^{7} + \cdots
\end{equation}
An alternate formula for $f(\tau)$ is as a generalised Eisenstein series\cite{Stein}
\[
f(\tau) = E_{2,\chi,\chi}(\tau):=\sum_{m=0}^\infty \Big[\sum_{n|m}\chi(n)\ \chi(\tfrac{m}n)\ m\Big] \ q^m\ ,
\]
where $\chi(m) = \left(\tfrac{-4}{m}\right)$ is a real Dirichlet character modulo $4$.
A basis for five-dimensional space of weight two modular forms of $\Gamma_0(16)$ is given by\footnote{We have used SAGE to obtain the dimension of the spaces of modular forms\cite{sage}. SAGE also provides a basis for the modular forms and we have verified that our choices are consistent with the choices given there.}
\begin{equation}
E^{(2)}_2(\tau),\  E^{(4)}_2(\tau),\  E^{(8)}_2(\tau),\  E^{(16)}_2(\tau) \textrm{ and } f(\tau)\ . 
\end{equation}
The first five Fourier coefficients of any weight two modular form of $\Gamma_0(16)$  uniquely determine the modular form.
A basis for weight two modular forms of $\Gamma_0(32)$ is obtained by adding three  more weight two modular
forms:
 \begin{equation*}
 E^{(32)}_2(\tau), f(2\tau)\   \textrm{ and the cusp form }\eta_{4^28^2}(\tau)\ ,
 \end{equation*}
to the $\Gamma_0(16)$ basis.
Eight of the first nine Fourier coefficients of any weight two modular form of $\Gamma_0(32)$  uniquely determines the modular form.

\subsection{Siegel and Jacobi Forms}

The group $Sp(2,\mathbb{Z})$ is the set of $4\times 4$ matrices written 
in terms of four $2\times 2$ matrices $A$, $B$, $C$,  $D$ (with integral entries)
as
$
M=\left(\begin{smallmatrix}
   A   & B   \\
    C  &  D
\end{smallmatrix}\right)
$
satisfying $ A B^T = B A^T $, $ CD^T=D C^T $ and $ AD^T-BC^T=I $. 
This group acts naturally 
on the Siegel upper half space, $\BH_2$, as
\begin{equation*}
\mathbf{Z}=\begin{pmatrix} \tau & z \\ z & \tau' \end{pmatrix}
\longmapsto M\cdot \mathbf{Z}\equiv (A \mathbf{Z} + B) 
(C\mathbf{Z} + D)^{-1} \ .
\end{equation*}
The  level $N\in \mathbb{Z}_{>0}$ subgroup, $\Gamma_0^{(2)}(N)$, of $Sp(2,\mathbb{Z})$ is given by restricting to $M$ such that $C=0\text{ mod }N$.

 The Jacobi sub-group, $\Gamma^J(N)\subset\Gamma_0^{(2)}(N)$,  is the semi-direct product of the Heisenberg group and $\Gamma_0(N)$ defined as follows:
     \begin{equation}
     \Gamma^J(N) = \Gamma_0(N)\ltimes H(\mathbb{Z}),
     \end{equation}
     where $H(\mathbb{Z})$ is given by
     \begin{align}
     &H(\mathbb{Z})=\left\{ [\lambda,\mu; \kappa]_H := \begin{pmatrix} 1& 0 & 0 & \mu\\ \lambda & 1 & \mu & \kappa\\ 0 & 0 & 1 & -\lambda\\ 0& 0& 0& 1 \end{pmatrix} , \hspace{0.5cm}  \lambda,\mu, \kappa \in \mathbb{Z}\right\} \ ,
     \end{align}
     The embedding of $\begin{pmatrix} a& b \\ c& d\end{pmatrix}\in \Gamma_0(N)$ in $\Gamma_t(N)$ is given by 
\begin{equation}
\widetilde{\begin{pmatrix}a& b \\ c& d\end{pmatrix}}\equiv \begin{pmatrix} a & 0 & b & 0\\ 0 & 1 & 0 & 0 \\ c & 0 & d & 0\\ 0 & 0 & 0 & 1\end{pmatrix}\ , \  c=0\mod N\ .
\end{equation}
Then, $\Gamma_0^{(2)}(N) $ is generated by adding the following transformation to $  \Gamma^J(N) $:
\begin{equation}
V =\begin{pmatrix} 0 & 1 & 0 & 0\\ 1 & 0 & 0 & 0 \\ 0 & 0 & 0 & 1\\ 0 & 0 & 1 & 0\end{pmatrix}\ ,
\end{equation}
with $\det(C\mathbf{Z}+D)=-1$. This acts on $\mathbb{H}_2$ as the involution
\begin{equation}
(\tau, z, \tau')\longrightarrow (\tau', z, \tau)\ .
\end{equation}

 A Siegel modular form, of weight $k$ with character $v$ with respect to $\Gamma_0^{(2)}(N)$, is a holomorphic function $F:\mathbb{H}_2\rightarrow \mathbb{C} $ satisfying
     \begin{equation}
     F|_k M (\mathbf{Z}) = v(M)\  F(\mathbf{Z})\ ,
     \end{equation}
     for all $M\in \Gamma_0^{(2)}(N) $ and the slash operation is defined as 
 \begin{equation}
    F|_k M(\mathbf{Z}) := \det (C\mathbf{Z}+D)^{-k}\ F(M\cdot \mathbf{Z})\ .
\end{equation}

\noindent \textbf{Jacobi Form:} A holomorphic function $\phi_{k,m}(\tau, z) : \mathbb{H}_1\times \mathbb{C}\rightarrow \mathbb{C}$ is called a  Jacobi form of weight $k$ and index $m$ if the function $$\tilde\phi_k(\mathbf{Z}) = \exp(2\pi i m \tau')\  \phi_{k,m}(\tau, z),$$ on $ \mathbb{H}_2$ such that $\tilde\phi_k(\mathbf{Z})$
 is a modular form of weight $k$ with respect to the Jacobi group $\Gamma^J(N)\subseteq Sp(2,\mathbb{Z})$ with character $v$ i.e., it satisfies
     \begin{equation}
     \tilde\phi|_k\ M (\mathbf{Z}) = v(M)\ \tilde\phi_k(\mathbf{Z})\hspace{1cm} \forall M\in \Gamma^J(N)\ ,
     \end{equation}
and it is holomorphic at all cusps, i.e.,
let $\gamma\in SL(2,\mathbb{Z})$ and $\tilde{\gamma}$, its embedding in $Sp(2,\mathbb{Z})$. Then it has a Fourier expansion 
     \begin{equation}
      \tilde\phi|_{k} \tilde{\gamma} (\mathbf{Z}) =  s^m\ \sum_{\substack{n,\ell\\ 4nm-l^2\ge0}}\ c_\gamma(n,\ell)\ q^n r^\ell =: s^m\ \sum_{\substack{n,\ell\\ \mathcal{D}=4nm-l^2\ge0}}\ C_\gamma(\mathcal{D},\ell)\ q^n r^\ell\ ,
     \end{equation}
     where $q= e^{2\pi i \tau}$, $r=e^{2\pi i z}$ and $s= e^{2\pi i \tau'}$ and $n,\ell\in\mathbb{Q}$. The combination $\mathcal{D}:=(4nm-\ell^2)$ is called the \textit{discrimant}, The Fourier coefficients of a Jacobi form depend only on the discriminant and the value of $\ell\mod 2m$.
     \begin{enumerate}[label=(\alph*)]
     \item $\phi_{k,m}(\tau,z)$ is called a \textit{cusp form} if $C_\gamma(\mathcal{D},\ell) = 0$ unless $\mathcal{D}>0$ at all cusps.
     \item $\phi_{k,m}(\tau,z)$ is called a \textit{weak Jacobi form} if $C_\gamma(\mathcal{D},\ell) = 0$ unless $n\ge0$ at all cusps.
     \item $\phi_{k,m}(\tau,z)$ is called a \textit{nearly holomorphic} if there exists an $n\in \mathbb{N}$ such that $\Delta^n\phi_{k,m}(\tau,z)$ is a weak Jacobi form where $\Delta =\eta(\tau)^{24}$. 
     \end{enumerate}
     The space of all Jacobi forms  with character for $\Gamma^J(N)$ is denoted by $J_{k,m}(\Gamma_0(N), v)$. Similarly the space for all weak and nearly holomorphic Jacobi forms are denoted by $J_{k,m}^w(\Gamma_0(N))$ and $J_{k,m}^{nh}(\Gamma_0(N))$ respectively.
     
 \subsection{Examples}
 
  An example of Jacobi form of weight $\frac12$ and index $\frac12$ is the Jacobi theta function of level 8:
     \begin{align}
     \theta_1 (\tau,z) &= \sum_{m\in\mathbb{Z}} \left(-\frac4m\right) q^{m^2/8} r^{m/2}\nonumber\\
	     &=-q^{1/8} r^{-1/2}\prod_{n\ge1} (1- q^{n-1}r) (1- q^{n}r^{-1}) (1- q^n)\ ,\\
             & =q^{1/8}\left(-\frac{1}{\sqrt{r}}+\sqrt{r}\right)+q^{9/8}
   \left(\frac{1}{r^{3/2}}-r^{3/2}\right)+ \cdots \nonumber
      \end{align}
      This is an element of $J_{\frac12, \frac12}(SL(2,\mathbb{Z}), v_\eta^3\times v_H)$ where $ v_\eta$ is the character of the Dedekind $\eta$-function and 
      \begin{equation}\label{Hcharacter}
       v_H([\lambda,\mu:\kappa]_H) =(-1)^{\lambda+\mu+\lambda\mu+\kappa} \ .
\end{equation}      
Then the weight $-1$ index $\frac12$ Jacobi form $\frac{\theta_1(\tau,z)}{\eta(\tau)^3}$ has character $v_H$.

 More generally, the genus-one theta functions are defined by
\begin{equation}
\theta\left[\genfrac{}{}{0pt}{}{a}{b}\right] \left(\tau,z\right)
=\sum_{l \in \BZ} 
q^{\frac12 (l+\frac{a}2)^2}\ 
r^{(l+\frac{a}2)}\ e^{i\pi lb}\ ,
\end{equation}
where $a,\ b\in (0,1)\mod 2$. We define $\theta_1 
\left(\tau,z\right)\equiv\theta\left[\genfrac{}{}{0pt}{}{1}{1}\right](\tau,z)$,
$\theta_2 
\left(\tau,z\right)\equiv\theta\left[\genfrac{}{}{0pt}{}{1}{0}\right] 
\left(z_1,z\right)$, $\theta_3 
\left(\tau,z\right)\equiv\theta\left[\genfrac{}{}{0pt}{}{0}{0}\right] 
\left(\tau,z\right)$ and $\theta_4 
\left(\tau,z\right)\equiv\theta\left[\genfrac{}{}{0pt}{}{0}{1}\right] 
\left(\tau,z\right)$.

The characters of the level 1 $\mathcal{N}=4$ superconformal algebra that appear in our decomposition of the  Jacobi forms of weight zero index 1 are 
$\mathcal{C}(\tau,z)$ and $q^{h-\frac18} \mathcal{B}(\tau,z)$ ($h>0$), where
\begin{align}
\mathcal{C}(\tau,z) &= \frac{\theta_1(\tau,z)^2}{\eta(\tau)^3} \frac{i}{\theta_1(\tau,2z)} 
\sum_{n\in\BZ} q^{2n^2} r^{4n} \frac{1+q^n r}{1-q^nr}\ . \label{masslesschar} \\
\mathcal{B}(\tau,z) 
 &= \frac{\theta_1(\tau,z)^2}{\eta(\tau)^3} \label{massivechar}\ .
\end{align}

\subsection{The EOT Jacobi Forms for $\Mtwo$}

Let $\phi_{0,1}(\tau,z)$ and $\phi_{-2,1}(\tau,z)$ denote the following Jacobi forms:
\begin{align*}
\phi_{0,1}(\tau,z) &= 4\left[\frac{\theta_2(\tau,z)^2}{\theta_2(\tau,0)^2} + \frac{\theta_3(\tau,z)^2}{\theta_3(\tau,0)^2}  + \frac{\theta_4(\tau,z)^2}{\theta_4(\tau,0)^2} \right]= (r^{-1} + 10 + r) + O(q)\ ,\\
\phi_{-2,1}(\tau,z) &=\frac{\theta_1(\tau,z)^2}{\eta(\tau)^6} =(r^{-1}-2 +r) + O(q)\ .
\end{align*}
These  are the unique weak Jacobi forms of index 1 and weight $\leq 0$. They generate the ring of weak Jacobi forms freely over the space of modular forms\cite{Dabholkar:2012nd}. Thus, all weak Jacobi forms of index 1 and weight zero can be given in terms of a constant and a modular form of weight 2 at suitable level.  Table \ref{EOTJFlist} lists all EOT Jacobi forms that appear for all conjugacy classes of $\Mtwo$. All EOT Jacobi forms have Fourier expansions about the cusp at $i\infty$ that are non-vanishing when the discriminant $\mathcal{D}\geq -1$. Further $C(\mathcal{D}-1,1)=2$ for all EOT Jacobi forms. This implies that terms with negative discriminant are identical for all of EOT Jacobi forms.
\begin{table}[ht]\label{EOTJFlist}
\begin{center}
\begin{tabular}{c|c|l}
Conj. Class & $\trho$ & \hspace{2in}$Z^\trho$ \\ \hline
1a & $1^{24}$ & 2 $\phi_{0,1}(\tau,z)$\phantom{\Big|}  \\[3pt]
2a/c & $2^{12}$ &  $(-2E^{(2)}_2(\tau)+4 E^{(4)}_2(\tau))\ \phi_{-2,1}(\tau,z) $  \\[3pt]
2a/c & $2^{12}$ &  $2 \eta_{1^82^{-4}}(\tau)\ \phi_{-2,1}(\tau,z) $  \\[3pt]
2b & $1^82^8$ & $\tfrac23 \phi_{0,1}(\tau,z) +\tfrac43 E^{(2)}_2(\tau)\ \phi_{-2,1}(\tau,z) $ \\[3pt]
3a & $1^63^6$ & $\tfrac12 \phi_{0,1}(\tau,z) +\tfrac32 E^{(3)}_2(\tau)\ \phi_{-2,1}(\tau,z) $ \\[3pt]
3b & $3^{8}$ &  $2\,\eta_{1^6 3^{-2}}(\tau)\ \phi_{-2,1}(\tau,z) $  \\[3pt]
4a & $1^4 2^2 4^6$ & $\tfrac13 \phi_{0,1}(\tau,z) +(-\tfrac13 E^{(2)}_2(\tau)+2 E^{(4)}_2(\tau))\ \phi_{-2,1}(\tau,z) $  \\[3pt]
5a & $1^45^4$ & $\tfrac13 \phi_{0,1}(\tau,z) +\tfrac53 E^{(5)}_2(\tau)\ \phi_{-2,1}(\tau,z) $ \\[3pt]
6a/c & $6^4$ & $2\,\eta_{ 1^2 2^23^{2}6^{-2}}(\tau)\ \phi_{-2,1}(\tau,z) $  \\[3pt]
6b & $1^22^23^26^2$ & $\tfrac16 \phi_{0,1}(\tau,z) +(-\tfrac16 E^{(2)}_2(\tau)-\tfrac12 E^{(3)}_2(\tau)+\tfrac52 E^{(6)}_2(\tau))\ \phi_{-2,1}(\tau,z) $  \\[3pt]
8a &$1^2 2^1 4^1 8^2$ & $\tfrac16 \phi_{0,1}(\tau,z) +(-\tfrac12 E^{(4)}_2(\tau)+\tfrac73 E^{(8)}_2(\tau))\ \phi_{-2,1}(\tau,z) $  \\[3pt]
10a/b/c & $2^2 10^2$ & 2\,$\eta_{ 1^3 2^15^{1}10^{-1}}(\tau)\ \phi_{-2,1}(\tau,z) $  \\[3pt]
11a &$1^2 11^2$ & $\tfrac16 \phi_{0,1}(\tau,z) + (\tfrac{11}6 E^{(11)}_2(\tau) - \tfrac{22}5 \eta_{1^211^2}(\tau))\ \phi_{-2,1}(\tau,z) $  \\[3pt]
4b & $2^4 4^4$&$2\,\eta_{ 2^{8}4^{-4}}(\tau)\ \phi_{-2,1}(\tau,z) $ \\[3pt]
4c & $4^6$ & $2\,\eta_{ 1^4 2^{2}4^{-2}}(\tau)\ \phi_{-2,1}(\tau,z) $\\[3pt]
12a & $12^2$ & $2\,\eta_{ 1^4 2^{-1}4^16^{1}12^{-1}}(\tau)\ \phi_{-2,1}(\tau,z) $ \\[3pt]
12b/c & $2^14^16^1 12^1$ &2$\,\eta_{ 1^3 2^{-1}3^{-1}4^26^{3}12^{-2}}(\tau)\ \phi_{-2,1}(\tau,z) $\\[3pt] \hline
\end{tabular}
\end{center}
\caption{The EOT Jacobi forms as given in \cite{Gaberdiel:2010ch,Gaberdiel:2010ca,Eguchi:2010fg} for all conjugacy classes of $M_{24}$ that reduce to conjugacy classes of $\Mtwo$.}
\end{table}

\clearpage

\section{Computations for the Borcherds product formula}\label{appendixcomp}

\subsection{Proving Equation \eqref{six-identities}}\label{six-appendix}

\subsubsection*{Cycle Shape $\hat\rho=2^6$}

The multiplicative seed, $\psi=\hat\psi_{0,1}^{2^6}(\tau,z)= \frac12 Z^{2^{12}}(\tau,z)$, is a Jacobi form of $\Gamma_0(4)$. The cusps and other data for $\Gamma_0(4)$ are as follows:
\begin{center}
\begin{tabular}{|c|c|c|c|}
\hline
$f/e$& $i\infty$ & $1/2$ & $0/1$\\
\hline
$h_e$ & 1&1&4\\
\hline
$N_e$& 1& 2&4\\
\hline
\end{tabular}
\end{center}
We need the Fourier expansion of $\psi$ about these cusps. Using $M_0=1$ and $M_{1/p}=-ST^{-p}S$, we can compute the expansion using methods described in chapter 2 of \cite{Niemann}. 
\begin{align}
\hat\psi^{\hat\rho}|_{0,1} M_0 (\tau,z)&=- 4 \frac{\eta(\tau)^8}{\eta(\tau/2)^4}\phi_{-2,1}(\tau,z)\\
\hat\psi^{\hat\rho}|_{0,1} M_{\frac12} (\tau,z)&=-  \frac{\eta(\tau)^8}{\eta(2\tau)^4}\phi_{-2,1}(\tau,z) = -\hat\psi_{0,1}^{\hat\rho}(\tau,z)
\end{align}
The Fourier expansion of the cusp about zero does not have any terms with integral powers of $q$ and hence this cusp does not contribute to the product formula. 

\subsubsection*{Cycle Shape $\hat\rho=3^4$}
The multiplicative seed, $\psi=\hat\psi_{0,1}^{3^4}(\tau,z)= \frac12 Z^{3^{8}}(\tau,z)$, is a Jacobi form of $\Gamma_0(9)$. The cusps and other data for $\Gamma_0(9)$ are as follows:
\begin{center}
\begin{tabular}{|c|c|c|c|c|}
\hline
$f/e$& $i\infty$ & $1/3$ & $2/3$ & $0/1$\\
\hline
$h_e$ & 1&1&1&9\\
\hline
$N_e$& 1& 3&3&9\\
\hline
\end{tabular}
\end{center}
Using $M_{2/3}=-S T^{-1} S T^2 S$, we obtain
\begin{align}
\hat\psi_{0,1}^{\hat\rho}(\tau,z)| M_{\frac13} &= e^{\frac{2\pi i}{3}}\hat\psi_{0,1}^{\hat\rho}(\tau,z)\\
\hat\psi_{0,1}^{\hat\rho}(\tau,z)| M_{\frac23} &=e^{\frac{4\pi i}{3}}\hat\psi_{0,1}^{\hat\rho}(\tau,z)
\end{align}
Again the cusp about $0$ does not have any terms with integral powers of $q$ and hence does not contribute to the Borcherds product.

\subsubsection*{Cycle Shape $\hat\rho=6^2$}
The multiplicative seed, $\psi=\hat\psi_{0,1}^{6^2}(\tau,z)= \frac12 Z^{6^{4}}(\tau,z)$, is a Jacobi form of $\Gamma_0(36)$. The cusps and other data for $\Gamma_0(36)$ are as follows:
\begin{center}
\begin{tabular}{|c|c|c|c|c|c|c|c|c|c|c|c|c|}
\hline
$f/e$& $i\infty$&0 & $1/2$ & $1/3$ & $2/3$ & $1/4$ & $1/6$ & $5/6$ & $1/9$ & $1/12$ & $5/12$ & $1/18$\\
\hline
$h_e$&1 & 36&9&4&4&9&1&1&4&1&1&1\\
\hline
$N_e$ & 1& 36&18&12&12&9&6&6&4&3&3&2\\
\hline
\end{tabular}
\end{center}
Using $M_{\frac56} = -ST^{-1}S T^5 S$ and $M_{\frac5{12}} = ST^{-2}S T^2 ST^{-2}S$,  we obtain
\begin{align*}
\hat\psi^{\hat\rho}|_{0,1}M_{\frac16}(\tau,z) &= -e^{\frac{-2\pi i}{3}}\hat\psi_{0,1}^{\hat\rho}(\tau,z)\\
\hat\psi^{\hat\rho}|_{0,1} M_{\frac56}(\tau,z)&= -e^{\frac{-4\pi i}{3}}\hat\psi_{0,1}^{\hat\rho}(\tau,z)\\
\hat\psi^{\hat\rho}|_{0,1} M_{\frac1{12}}(\tau,z) &=e^{\frac{2\pi i}{3}}\hat\psi_{0,1}^{\hat\rho}(\tau,z)\\
\hat\psi^{\hat\rho}|_{0,1} M_{\frac5{12}}(\tau,z) &=e^{\frac{4\pi i}{3}}\hat\psi_{0,1}^{\hat\rho}(\tau,z)\\
\hat\psi^{\hat\rho}|_{0,1}M_{\frac1{18}}(\tau,z) &=-\hat\psi_{0,1}^{\hat\rho}(\tau,z)
\end{align*}
The Fourier expansions about the cusps at $0, 1/2, 1/3, 2/3, 1/4, 1/9$ do not have any terms with integral powers of $q$ and hence do not contribute to the Borcherds product.

\section{Basic Group Theory}

\subsection{$M_{12}$ and  $M_{12}\!:\!2$}
In the 12-dimensional permutation representation, $M_{12}$ is generated as  $\langle \alpha, \beta, \gamma,\delta\rangle$. 
$L_2(11)_A$ is a maximal subgroup of $M_{12}$ while one has $L_2(11)_B \subset M_{11} \subset M_{12}$.
The four conjugacy classes associated with elements of order $4$, $8$ and $10$ do \textit{not} reduce to conjugacy classes of
either $L_2(11)_A $ or $L_2(11)_B$.

Let $g$ denote an element of $M_{12}$ in the 12-dimensional permutation representation. Let $\varphi$ denote the outer automorphism of $M_{12}$.
 It acts on the generators of $M_{12}$ as
\[
\alpha^\varphi = \varphi \alpha \varphi^{-1}=\alpha^{-1} \quad,\quad \beta^\varphi =\beta \quad,\quad
\gamma^\varphi=\gamma^{-1}\quad,\quad \delta^\varphi =\delta\ .
\]
The 24-dimensional permutation representation of $\Mtwo$
consists of two classes of elements given in block-diagonal form below:
\[
(g,e):= \begin{pmatrix} g & 0 \\ 0 & \varphi(g) \end{pmatrix} \quad\text{and}\quad 
(g,\varphi):= \begin{pmatrix} 0 & g \\ \varphi(g) & 0 \end{pmatrix}\ .
\]
Conjugacy classes, $\widetilde{\rho}$ of type $(g,e)$ of $\Mtwo$ descend to pairs of conjugacy classes of $M_{12}$. Explicitly, one has $\widetilde{\rho}=(\wrho,\varphi(\wrho))$, where $\wrho$ is the conjugacy class of $g$ in $M_{12}$. One has the sequence of groups 
\[
L_2(11)_{A/B} \subset M_{12} \stackrel{\varphi}{\longrightarrow} M_{12}\!:\!2 \subset M_{24}\ ,
\]

\subsubsection{$\Mtwo$ characters from $M_{12}$ characters}

Given a group $G$ and a $\BZ_2$ automorphism $\varphi$, the characters  of the group $G$ are related to those of the group $G.2$ in two possible ways\cite{Atlasv3}:
\begin{enumerate}
\item \textit{The splitting case:}
 A character $\hat{\chi}_m$ of $G$ may give rise to two characters of $G.2$ -- call them $\tilde{\chi}_a$ and $\tilde{\chi}_{a'}$. For elements of type $(g,e)$, they are given by
$\tilde{\chi}_a=\tilde{\chi}_{a'}= \hat{\chi}_m$. For elements of type $(g,\varphi)$, one has $\tilde{\chi}_a+\tilde{\chi}_{a'}=0$. Thus we have a natural pairing $(a,a')$ of representations of $G.2$ and they are mapped to the $m$-th representation of $G$. For $\Mtwo$, the splitting representations are (in the notation $m\leftrightarrow (a,a')$)
\begin{equation}
\begin{aligned}
1 \leftrightarrow (1,2)\ , \ 
6 \leftrightarrow (5,6) \ , \ 
7 \leftrightarrow (7,8) \ , \
8 \leftrightarrow (9,10) \ , \\
11+a \leftrightarrow (12+2a,13+2a) \text{ for } a=0,1,\ldots, 4\ .
\end{aligned}
\end{equation}
\item \textit{The fusion case:} Two characters $\hat{\chi}_m$ and $\hat{\chi}_n$ fuse to give a single character, call it $\tilde{\chi}_{m,n}$. For elements of $G.2$ of type $(g,e)$, one has
$\tilde{\chi}_a=\hat{\chi}_m+\hat{\chi}_n$  and for elements of type $(g,\varphi)$, one has $\tilde{\chi}_a=0$. 

The fusion characters of $\Mtwo$ are (in the notation $(m,n)\leftrightarrow a$)
\begin{equation}
(2,3) \leftrightarrow 3 \ , \
(4,5) \leftrightarrow 4 \ , \
(9,10) \leftrightarrow 11 \ . 
\end{equation}

\end{enumerate}

\subsection{Character Tables}

Character table for $L_2(11)$ obtained from the GAP database\cite{GAP4}

\[
\begin{array}{c|rrrccrcc}
&1a & 2a &3a& 5a &5b &6a &11a & 11b \\[3pt] \hline
 \chi_1& 1 & 1 & 1 & 1 & 1 & 1 & 1 & 1 \\
 \chi_2& 5 & 1 & -1 & 0 & 0 & 1 & -\frac{1}{2}+\frac{i \sqrt{11}}{2} & -\frac{1}{2}-\frac{i \sqrt{11}}{2}
   \\
  \chi_3&5 & 1 & -1 & 0 & 0 & 1 & -\frac{1}{2}-\frac{i \sqrt{11}}{2} & -\frac{1}{2}+\frac{i \sqrt{11}}{2}
   \\
  \chi_4&10 & -2 & 1 & 0 & 0 & 1 & -1 & -1 \\
  \chi_5&10 & 2 & 1 & 0 & 0 & -1 & -1 & -1 \\
  \chi_6&11 & -1 & -1 & 1 & 1 & -1 & 0 & 0 \\
  \chi_7&12 & 0 & 0 & -\frac{1}{2}+\frac{\sqrt{5}}{2} & -\frac{1}{2}-\frac{\sqrt{5}}{2} & 0 & 1 & 1 \\
  \chi_8&12 & 0 & 0 & -\frac{1}{2}-\frac{\sqrt{5}}{2} & -\frac{1}{2}+\frac{\sqrt{5}}{2} & 0 & 1 & 1 \\[3pt] \hline
\end{array}
\]


The character table for $M_{12}$ (obtained from the GAP character table database)
\begin{equation}
\left(\begin{smallmatrix}
\textrm{Label}         &1a &  2a & 2b & 3a &  3b &  4a & 4b & 5a & 6a & 6b & 8a & 8b & 10a &  11a & 11b \\[4pt]
\widehat{\chi}_1      &~1  &~1  &~1  &~1  &~1  &~1  &~1  &~1  &~1  &~1  &~1  &~1   &~1   &~1   &~1 \\
\widehat{\chi}_2   &   11& -1&  ~3&  ~2 &-1 &-1 & ~3 & 1 &-1  &~0 &-1 & ~1 & -1   &~0   &~0 \\
\widehat{\chi}_3    &  11'& -1 & ~3 & ~2& -1&  ~3& -1& ~1& -1  &~0  &~1 &-1 & -1   &~0   &~0 \\
\widehat{\chi}_4    &  16&  ~4  &~0& -2  &~1  &~0  &~0 & ~1 & ~1  &~0  &~0  &~0 & -1  & ~\alpha & \alpha^* \\
\widehat{\chi}_5    &  16' & ~4  &~0 &-2&  ~1  &~0  &~0  &~1 & ~1  &~0  &~0  &~0  &-1 & \alpha^* &  ~\alpha\\
\widehat{\chi}_6    &  45 & ~5 &-3  &~0 & ~3&  ~1 & ~1  &~0 &-1  &~0 &-1 &-1   &~0 &  ~1&   ~1\\
\widehat{\chi}_7  &    54 & ~6  &~6  &~0  &~0 & ~2 & ~2& -1  &~0  &~0  &~0  &~0  & ~1  &-1  &-1\\
\widehat{\chi}_8  &    55_R &-5  &~7  &~1 & ~1 &-1& -1  &~0 & ~1 & ~1 &-1& -1   &~0   &~0   &~0 \\ 
\widehat{\chi}_9  &    55& -5 &-1  &~1 & ~1&  ~3& -1  &~0 & ~1 &-1 &-1 & ~1   &~0   &~0   &~0 \\ 
\widehat{\chi}_{10}&     55'& -5 &-1&  ~1 & ~1 &-1 & ~3  &~0 & ~1& -1&  ~1 &-1   &~0   &~0   &~0\\
\widehat{\chi}_{11} &   66 & ~6  &~2&  ~3  &~0 &-2& -2&  ~1  &~0& -1  &~0  &~0 & ~1   &~0   &~0\\
\widehat{\chi}_{12} &    99& -1&  ~3  &~0 & ~3 &-1 &-1 &-1& -1  &~0 & ~1 & ~1 & -1   &~0   &~0\\
\widehat{\chi}_{13}  &  120  &~0 &-8&  ~3  &~0  &~0  &~0  &~0  &~0 & ~1  &~0  &~0   &~0 & -1  &-1\\
\widehat{\chi}_{14} &  144 & ~4  &~0  &~0 &-3  &~0  &~0& -1 & ~1  &~0  &~0  &~0 & -1&   ~1 &  ~1\\
\widehat{\chi}_{15}&    176 & -4  &~0& -4 &-1  &~0  &~0 & ~1& -1  &~0  &~0  &~0 &  ~1   &~0   &~0\\
\end{smallmatrix}
\right)\label{M12chartable}
\end{equation}
where $\alpha=-\tfrac12 + i \tfrac{\sqrt{11}}2$. Under the outer automorphism, $\varphi$, of $M_{12}$ one has
\begin{equation}
\varphi:\quad \wchi_2 \leftrightarrow \wchi_3\quad,\quad
\wchi_4 \leftrightarrow \wchi_5\quad,\quad
\wchi_9 \leftrightarrow \wchi_{10}\ .
\end{equation}

The character table for $\Mtwo$ (obtained from the GAP  database)
\begin{equation*}
\left(\begin{smallmatrix}
\textrm{Label} & 1a & 2a & 2b & 3a & 3b & 4a & 5a &  6a & 6b & 8a & 10a & 11a & 2c & 4b & 4c & 6c & 10b & 10c & 12a & 12b & 12c\\[3pt]
\widetilde{\chi}_1 & ~1 & ~1 & ~1 & ~1 & ~1 & ~1 & ~1 & ~1 & ~1 & ~1&1&1&~1 & ~1 & ~1 & ~1 & ~1 & ~1 & ~1 & ~1 & ~1 \\
\widetilde{\chi}_2 & ~1 & ~1 & ~1 & ~1 & ~1 & ~1 & ~1 & ~1 & ~1 & ~1 & ~1 & ~1 & -1 & -1 & -1 & -1 & -1 & -1 & -1 & -1 & -1 \\
 \widetilde{\chi}_3 &22 & -2 & ~6 & ~4 & -2 & ~2 & ~2 & -2 & ~0 & ~0 & -2 & ~0 & ~0 & ~0 & ~0 & ~0 & ~0 &
   ~0 & ~0 & ~0 & ~0 \\
\widetilde{\chi}_4 & 32 & ~8 & ~0 & -4 & ~2 & ~0 & ~2 & ~2 & ~0 & ~0 & -2 & -1 & ~0 & ~0 & ~0 & ~0 & ~0 &
   ~0 & ~0 & ~0 & ~0 \\
\widetilde{\chi}_5 & 45 & ~5 & -3 & ~0 & ~3 & ~1 & ~0 & -1 & ~0 & -1 & ~0 & ~1 & ~5 & -3 & ~1 & -1 & ~0
   & ~0 & ~1 & ~0 & ~0 \\
\widetilde{\chi}_6 & 45 & ~5 & -3 & ~0 & ~3 & ~1 & ~0 & -1 & ~0 & -1 & ~0 & ~1 & -5 & ~3 & -1 & ~1 & ~0
   & ~0 & -1 & ~0 & ~0 \\
 \widetilde{\chi}_7 &54 & ~6 & ~6 & ~0 & ~0 & ~2 & -1 & ~0 & ~0 & ~0 & ~1 & -1 & ~0 & ~0 & ~0 & ~0 &
   ~\sqrt{5} & -\sqrt{5} & ~0 & ~0 & ~0 \\
 \widetilde{\chi}_8 &54 & ~6 &~ 6 & ~0 & ~0 & ~2 & -1 & ~0 & ~0 &~ 0 & ~1 & -1 & ~0 & ~0 & ~0 & ~0 &
   -\sqrt{5} & ~\sqrt{5} & ~0 & ~0 & ~0 \\
 \widetilde{\chi}_9 &55 & -5 & ~7 & ~1 & 1 & -1 & ~0 & ~1 & ~1 & -1 & ~0 & ~0 & ~5 & ~1 & -1 & -1 & 0
   & ~0 & -1 & ~1 & ~1 \\
 \widetilde{\chi}_{10} &55 & -5 & ~7 & ~1 & 1 & -1 & ~0 & ~1 & ~1 & -1 & ~0 & ~0 & -5 & -1 & ~1 & ~1 & ~0
   & ~0 & ~1 & -1 & -1 \\
 \widetilde{\chi}_{11} & 110 & -10 & -2 & ~2 & ~2 & ~2 & ~0 & ~2 & -2 & ~0 & ~0 & ~0 & ~0 & ~0 & ~0 & ~0 & ~0
   & ~0 & ~0 & ~0 & ~0 \\
 \widetilde{\chi}_{12} &66 & ~6 & ~2 & ~3 & ~0 & -2 & ~1 & ~0 & -1 & ~0 & ~1 & ~0 & ~6 & ~2 & ~0 & ~0 & ~1 & ~1
   & ~0 & -1 & -1 \\
 \widetilde{\chi}_{13}& 66 & ~6 & ~2 & ~3 & ~0 & -2 & ~1 & ~0 & -1 & ~0 & ~1 & ~0 & -6 & -2 & ~0 & ~0 & -1
   & -1 & ~0 & ~1 & ~1 \\
  \widetilde{\chi}_{14}&99 & -1 & ~3 & ~0 & ~3 & -1 & -1 & -1 & ~0 & ~1 & -1 & ~0 & ~1 & -3 & -1 & ~1 &
   ~1 & ~1 & -1 & ~0 & ~0 \\
 \widetilde{\chi}_{15} &99 & -1 & ~3 & ~0 & ~3 & -1 & -1 & -1 & ~0 & ~1 & -1 & ~0 & -1 & ~3 & ~1 & -1 &
   -1 & -1 & ~1 & ~0 &~ 0 \\
 \widetilde{\chi}_{16} &120 & ~0 & -8 & ~3 & ~0 & ~0 & ~0 & ~0 & ~1 & ~0 & ~0 & -1 & ~0 & ~0 & ~0 & ~0 & ~0 &
   ~0 & ~0 & ~\sqrt{3} & -\sqrt{3} \\
 \widetilde{\chi}_{17} &120 & ~0 & -8 & ~3 & ~0 & ~0 & ~0 & ~0 & ~1 & ~0 & ~0 & -1 & ~0 & ~0 & ~0 & ~0 & ~0 &
   ~0 & ~0 & -\sqrt{3} & ~\sqrt{3} \\
  \widetilde{\chi}_{18}&144 & ~4 & ~0 & ~0 & -3 & ~0 & -1 & ~1 & ~0 & ~0 & -1 & ~1 & ~4 & ~0 & ~2 & ~1 & -1
   & -1 & -1 & ~0 & ~0 \\
 \widetilde{\chi}_{19} &144 & ~4 & ~0 & ~0 & -3 & ~0 & -1 & ~1 & ~0 & ~0 & -1 & ~1 & -4 & ~0 & -2 & -1 &
   ~1 & ~1 & ~1 & ~0 & ~0 \\
 \widetilde{\chi}_{20} &176 & -4 & ~0 & -4 & -1 & ~0 & ~1 & -1 & ~0 & ~0 & ~1 & ~0 & ~4 & ~0 & -2 & ~1 &
   -1 & -1 & ~1 & ~0 & ~0 \\
 \widetilde{\chi}_{21} &176 & -4 & ~0 & -4 & -1 & ~0 & ~1 & -1 & ~0 & ~0 & ~1 & ~0 & -4 & ~0 & ~2 & -1 &
   ~1 & ~1 & -1 & ~0 & ~0
\end{smallmatrix}
\right)\ .
\end{equation*}

\begin{sidewaystable}
\scriptsize
\begin{tabular}{c|rrrrrrrrrrrrrrr}
 & $\wchi _1 $& $\wchi _2$ & $\wchi _3$ &$ \wchi _4$ &$ \wchi _5$ &$ \wchi _6$ &$ \wchi _7 $&
   $\wchi _8$ & $\wchi _9$ & $\wchi _{10}$ & $\wchi _{11}$ & $\wchi _{12} $& $\wchi _{13}$
   & $\wchi _{14}$ & $\wchi _{15}$ \\[3pt] \hline
$ 1$ & $-1$ & 0 & 0 & 0 & 0 & 0 & 0 & 0 & 0 & 0 & 0 & 0 & 0 & 0 & 0 \\
 $ q$  & 0 & 0 & 0 & 0 & 0 & 1 & 0 & 0 & 0 & 0 & 0 & 0 & 0 & 0 & 0 \\
 $ q^2$  & 0 & 0 & 0 & 0 & 0 & 0 & 0 & 1 & 0 & 0 & 0 & 0 & 0 & 0 & 1 \\
 $ q^3$  & 0 & 0 & 0 & 0 & 0 & 0 & 0 & 0 & 0 & 0 & 1 & 0 & 2 & 2 & 1 \\
 $ q^4$  & 0 & 0 & 0 & 0 & 0 & 1 & 2 & 3 & 0 & 4 & 1 & 3 & 2 & 3 & 4 \\
 $ q^5$  & 0 & 0 & 0 & 2 & 2 & 4 & 4 & 1 & 3 & 3 & 4 & 5 & 8 & 9 & 11 \\
 $ q^6$  & 0 & 2 & 4 & 1 & 1 & 4 & 8 & 10 & 8 & 10 & 10 & 15 & 16 & 21 & 26 \\
 $ q^7$  & 0 & 0 & 4 & 7 & 7 & 18 & 16 & 15 & 24 & 10 & 23 & 32 & 42 & 46 & 56 \\
 $ q^8$  & 1 & 6 & 12 & 10 & 10 & 28 & 38 & 43 & 46 & 32 & 43 & 70 & 78 & 98 & 124 \\
 $ q^9$  & 1 & 13 & 15 & 23 & 23 & 66 & 76 & 70 & 113 & 37 & 94 & 134 & 174 & 206 & 242 \\
 $ q^{10}$  & 3 & 31 & 35 & 42 & 42 & 119 & 148 & 162 & 219 & 89 & 179 & 276 & 322 & 390 & 485 \\
 $ q^{11}$  & 4 & 62 & 40 & 88 & 88 & 242 & 278 & 272 & 442 & 122 & 346 & 511 & 632 & 753 & 914 \\
 $ q^{12}$  & 10 & 146 & 84 & 147 & 147 & 420 & 522 & 546 & 809 & 259 & 633 & 956 & 1144 & 1384 & 1699 \\
 $ q^{13}$  & 19 & 264 & 102 & 286 & 286 & 801 & 938 & 933 & 1506 & 396 & 1152 & 1716 & 2102 & 2506 & 3051 \\
 $ q^{14}$  & 30 & 500 & 192 & 484 & 484 & 1364 & 1664 & 1721 & 2628 & 768 & 2018 & 3056 & 3666 & 4420 & 5423 \\
 $ q^{15}$  & 52 & 889 & 263 & 861 & 861 & 2420 & 2874 & 2896 & 4603 & 1241 & 3535 & 5263 & 6434 & 7697 & 9375 \\
 $ q^{16}$  & 94 & 1579 & 455 & 1444 & 1444 & 4069 & 4922 & 5058 & 7754 & 2290 & 5994 & 9033 & 10886 & 13087 & 16032 \\
 $ q^{17}$  & 151 & 2664 & 652 & 2468 & 2468 & 6920 & 8248 & 8340 & 13003 & 3773 & 10099 & 15107 & 18382 & 22027 & 26887 \\
 $ q^{18}$  & 252 & 4501 & 1133 & 4020 & 4020 & 11330 & 13674 & 14000 & 21243 & 6639 & 16689 & 25077 & 30316 & 36427 & 44563 \\
 $ q^{19}$  & 412 & 7330 & 1686 & 6647 & 6647 & 18681 & 22316 & 22644 & 34484 & 10950 & 27318 & 40913 & 49696 & 59567 & 72744 \\
 $ q^{20}$  & 669 & 11917 & 2853 & 10649 & 10649 & 29960 & 36064 & 36844 & 54889 & 18611 & 44021 & 66134 & 80010 & 96094 & 117541 \\
 $ q^{21}$  & 1064 & 18925 & 4427 & 17087 & 17087 & 48040 & 57526 & 58442 & 86807 & 30313 & 70371 & 105420 & 127988 & 153496 & 187481 \\
 $ q^{22}$  & 1692 & 29831 & 7327 & 26877 & 26877 & 75625 & 90908 & 92775 & 135259 & 50001 & 111037 & 166710 & 201830 & 242298 & 296284 \\
 $ q^{23}$  & 2622 & 46244 & 11482 & 42197 & 42197 & 118616 & 142120 & 144536 & 209293 & 80135 & 173798 & 260529 & 316064 & 379145 & 463254 \\
 $ q^{24}$  & 4082 & 71296 & 18694 & 65174 & 65174 & 183384 & 220348 & 224690 & 320080 & 128864 & 269200 & 403992 & 489368 & 587424 & 718126 \\
 $ q^{25}$  & 6270 & 108377 & 29259 & 100406 & 100406 & 282327 & 338446 & 344382 & 486339 & 202971 & 413792 & 620437 & 752450 & 902705 & 1103084 \\
 $ q^{26}$  & 9555 & 163767 & 46683 & 152718 & 152718 & 429576 & 515886 & 525845 & 731812 & 319208 & 630341 & 945863 & 1145966 & 1375439 & 1681406 \\
 $ q^{27}$  & 14433 & 244901 & 72561 & 231277 & 231277 & 650388 & 780008 & 793968 & 1094465 & 494269 & 953589 & 1429925 & 1733926 & 2080389 & 2542299 \\
 $ q^{28}$  & 21711 & 364030 & 113550 & 346819 & 346819 & 975551 & 1171218 & 1193511 & 1623580 & 762466 & 1431222 & 2147351 & 2602046 & 3122821 & 3817239 \\
 $ q^{29}$  & 32314 & 536411 & 174379 & 517616 & 517616 & 1455614 & 1746034 & 1777621 & 2394700 & 1161740 & 2134316 & 3200923 & 3880816 & 4656537 & 5690817 \\
 $ q^{30}$  & 47909 & 786171 & 268275 & 766024 & 766024 & 2154660 & 2586488 & 2635260 & 3507492 & 1761600 & 3160915 & 4742013 & 5746832 & 6896777 & 8429971 \\
$  q^{31}$  & 70489 & 1143629 & 406567 & 1128391 & 1128391 & 3173388 & 3806978 & 3876453 & 5110133 & 2644483 & 4653450 & 6979350 & 8461124 & 10152616 & 12408027 \\
$ q^{32}$  & 103184 & 1655300 & 615246 & 1650225 & 1650225 & 4641456 & 5570988 & 5675403 & 7399554 & 3949188 & 6808400 & 10213676 & 12378560 & 14855132 & 18157233 \\
\end{tabular}
\caption{Character Decomposition of the $M_{12}$ Jacobi forms.}\label{characterdecomp}
\end{sidewaystable}
\bibliography{master}
\end{document}